\documentclass{ecai}
\usepackage{graphicx}
\usepackage{latexsym}

\ecaisubmission  
\usepackage{amsfonts}
\usepackage{amssymb}
\usepackage{amsmath}
\usepackage{amsthm}
\usepackage[square,numbers]{natbib}
\usepackage{xcolor}
\graphicspath{ {../images/} }

\newcommand{\gp}[1]{\textcolor{blue}{gp: #1}}

\usepackage{url}
\usepackage[hidelinks]{hyperref}
\usepackage{booktabs}
\newtheorem{thm}{Theorem}

\usepackage{algorithm,algorithmic}
\makeatletter
\renewcommand{\ALG@name}{Procedure}
\renewcommand\fnum@algorithm{\fname@algorithm~\thealgorithm.}
\makeatother
\makeatletter
\let\sv@thm\@thm
\def\@thm{\let\indent\relax\sv@thm}
\makeatother

\makeatletter
\def\old@comma{,}
\catcode`\,=13
\def,{%
  \ifmmode%
    \old@comma\discretionary{}{}{}%
  \else%
    \old@comma%
  \fi%
}
\makeatother
\usepackage{wrapfig}
\usepackage{changepage}
\usepackage{balance} 
\begin{document}

\begin{frontmatter}

\title{As Time Goes By: Adding a Temporal Dimension\\ Towards Resolving Delegations in Liquid Democracy}

% \author[]{Paper ID: 1997}
\author[A,B]{\fnms{Evangelos}~\snm{Markakis}
% \orcid{....-....-....-....}
}
\author[A]{\fnms{Georgios}~\snm{Papasotiropoulos}
% \orcid{....-....-....-....}
} % use of \orcid{} is optional

\address[A]{\small Athens University of Economics and Business, Athens, Greece}
\address[B]{\small Input Output Global (IOG)}

\begin{abstract}
{\footnotesize{
In recent years, the study of various models and questions related to Liquid Democracy has been of growing interest among the community of Computational Social Choice. A concern that has been raised, is that current academic literature focuses solely on static inputs, concealing a key characteristic of Liquid Democracy: the right for a voter to change her mind as time goes by, regarding her options of whether to vote herself or delegate her vote to other participants, till the final voting deadline. In real life, a period of extended deliberation preceding the election-day motivates voters to adapt their behaviour over time, either based on observations of the remaining electorate or on information acquired for the topic at hand. By adding a temporal dimension to Liquid Democracy, such adaptations can increase the number of possible delegation paths and reduce the loss of votes due to delegation cycles or delegating paths towards abstaining agents, ultimately enhancing participation. Our work takes a first step to integrate a time horizon into decision-making problems in Liquid Democracy systems. Our approach, via a computational complexity analysis, exploits concepts and tools from temporal graph theory which turn out to be convenient for our framework.}}

\end{abstract}

\end{frontmatter}

\section{Introduction}
\label{sec:intro}
Liquid Democracy (LD) is a novel voting framework that aspires to revolutionize the typical voter's perception of civic engagement and ultimately elevate both the quantity and quality of community involvement. At its core, LD is predicated on empowering voters to determine their mode of participation. This can be achieved by either casting a vote directly, as in direct democracy, or by entrusting a proxy to act on their behalf, as in representative democracy. Notably, delegations are transitive, meaning that a delegate's vote can be delegated afresh, and at the end of the day a voter that has decided to cast a ballot, votes with a weight dependent on the number of agents that she
represents, herself included. As a result of its flexibility, LD is alleged to reconcile the appeal of direct democracy with the practicality of representative democracy, yielding the best of both worlds. The origin of the ``liquid'' metaphor remains a matter of debate up to date, with one view being that it stems from the ability of votes to flow along delegation paths, while an alternative view argues that it arises from the ability of voters to revoke delegation approvals and continuously adjust their choices. 
 As we will justify shortly after, current work tends to forcefully support the second opinion.

According to \cite{Zuber} there is a number of features that suffice to establish a framework as a Liquid Democracy one. Most of them are related to the transitivity property and to the options given to the voters about casting a ballot or choosing representatives. These are more or less taken into account in all relevant works that come from the field of Computational Social Choice. A further aspect, 
called \textit{Instant Recall}, encompasses the ability of voters to withdraw their delegation at any time. As a matter of fact, in practice, elections allow for extended (sometimes structured) periods of deliberation, until the votes are finalized, and Liquid Democracy could serve as a means of debate empowerment. A revocation of delegation may occur due to disagreements with a delegate's post-delegation choices, doubts on the integrity of one's behavior, or an agent's further understanding of the issue under consideration. A characteristic that is being shared by all the works in the AI community is that they all seem to ignore the Instant Recall feature, and examine isolated static delegation profiles. This oversight was identified and criticized by the team behind the LiquidFeedback platform \cite{liquidfeedback}, the most influential and large scale experiment of LD. In \cite{Nitsche}, inter alia, they claim the following:

\begin{adjustwidth}{0.27cm}{0.27cm}
\textit{In a governance system with a continuous stream of decisions, we expect that participants observe the actions (and even non-actions) of other participants, in particular the activities of their (direct and indirect) delegates as well as the activities of other participants, who they consider as delegates. Based on their observations, we expect participants to adapt their own behaviour in respect to setting, changing, and removing delegations and their own participation. Based on the track records of the participants, a network of trust or dynamic scheme of representation proves itself to be a responsible power structure. [...] We believe that the effects that occur through observation and adaptation over time are an essential prerequisite for a comprehensive understanding of liquid democracy, (which) requires a broader view, namely adding a temporal dimension to delegation models.}
\end{adjustwidth}

Leaving aside the lack of temporal aspects in the literature, there are also additional concerns to address in traditional LD models. A crucial disadvantage is that we may experience delegation cycles or delegation paths towards abstainers, which result to inevitably lost votes. A way that has been suggested in theory \cite{goelz,kotsialou2018incentivising,brill2022liquid} and has been implemented in practice \cite{google}, in order to mitigate such issues is to allow each delegating agent to specify an entire set of agents she approves as potential representatives together with a ranking among them that indicates her preferences. Nevertheless, even with these efforts, the discussed issues may still arise at the election-day.
% However, even with these efforts, the problems of delegation paths towards abstainers and delegation cycles, at the election-day, can still arise. 
And here is where the temporal dimension can come into play! The main focus of our work is in proposing a framework that leverages temporal information to address the identified concerns, while also providing a valuable tool for deliberation.
%that generalizes previous attempts and that utilizes all available information at the time of the election, to confront cases where 
% identified above.

%When multiple approved representatives are permitted per voter, many different feasible delegation paths towards voters who choose to cast a ballot may emerge. Therefore, a proper method must be selected to choose between them. {\it Delegation rules} are the centralized procedures designed specifically for this purpose. To assess the quality of a delegation rule we use the criterion of maximizing the total utility of the electorate, as also done in \cite{escoffier} and \cite{markakis2021approval}. The obvious requirement for a delegation rule, created with the prospect of being implemented in practice is to be computationally \textit{efficient} and therefore, our focus is on studying the computational complexity of designing delegation rules that satisfy certain desirable properties. A discussion on the properties for delegation rules, that our work focuses on, follows.

In particular, our aim is to study the existence of efficient \textit{delegation rules} that fulfill certain desirable axioms. A delegation rule is a centralized algorithm that takes as input the available information of the deliberation phase and prescribes for each non-abstaining participant a delegation path to a voter who casts a ballot. The main properties that we would like our rules to satisfy are described below.  

\smallskip \noindent {\bf Time-Conscious Delegation Rules.} We view the temporal dimension as an important feature in the design of delegation rules. To demonstrate this, consider an election where the information at the very end of the deliberation phase can only produce a path to a cycle or to an abstainer, for some voter. Our main insight is that one way to resolve such scenarios 
%and assign a path for $v$ to a voter who actually casts a ballot 
is to look into approvals expressed during the previous time steps of the deliberation phase. Our work operates under the premise that if a voter $v$ decides to trust another voter $u$, at a given moment in time, say $t$, then $v$ accepts any decisions made by $u$ at time $t$ or earlier (up to a certain number of time steps prior to $t$, which could be given as a parameter by voter $v$). This is because the decision to approve a delegation to $u$ is based on what $v$ observes in the previous time steps and up until time $t$. However, voter $v$ still retains the right to revoke her approval to $u$ at a later point in time. If this occurs, then voter $u$ is permitted to represent $v$ only if she chooses an action that she had declared at or before time $t$. We refer to the rules that produce delegation paths respecting in such a way the ordering of the time-instants at which a delegation is made available, as \textit{time-conscious} (for a formal definition refer to Section \ref{sec:prelims1}). In our model, time-conscious delegation rules can guarantee the absolute approval of a delegating voter to her ultimate representative. 
%That is because, non-time-conscious delegation rules risk result in the situation where an agent casts a vote on behalf of another without their current approval. 
%On the contrary, a non-time-conscious delegation rule may prompt a delegate to transfer their whole voting power to another voter at a time $t$, even if some of the voters that she represents revoked their approval of the delegate prior to $t$.

\smallskip \noindent {\bf Confluent Delegation Rules.} In models incorporating multiple, ranked, delegations, as the one under consideration, an esteemed property is \textit{confluence}, which posits that each voter should have at most one other immediate representative in the final outcome \cite{brill2022liquid}. This desirable attribute guarantees that every voter is instructed to take a single action among the three options: vote, abstain, or delegate her own and all received ballots to a specific voter. On the contrary, a non-confluent rule may %lead to complicated and confusing instructions for a voter during the voting phase: for example, a voter might be 
prompt a voter to delegate different ballots received from different voters to different representatives (and even delegate her own ballot to yet another representative). Such suggestions can be challenging for a voter to follow. In addition to its intuitive nature, confluence is also significant for maintaining transparency and preserving the high level of accountability inherent in classical Liquid Democracy, as highlighted in \cite{goelz}.

\subsection{Contribution}
Conceptually, we view as our main contribution that we explicitly add a temporal dimension to (a generalization of) an existing framework. Hence this is putting a stake in the ground in bridging a significant research gap identified by practitioners. 
We then study the compatibility of computational tractability with desirable properties of delegation rules, with the objective of reducing the loss of votes resulting from delegation cycles or paths towards abstaining agents, and ultimately enhance the electorate's participation.
% We study the compatibility of computational tractability with desirable properties of delegation rules. 
Namely we are interested in polynomially computable rules that maximize the total utility of the electorate and at the same are time-conscious and confluent. Unfortunately, despite the natural appeal of these requirements, it turns out that this is too much to ask for: our results demonstrate that such a delegation rule does not exist, unless P=NP, even for simple variants of our model.
Therefore, the best one could hope for is to design efficient procedures that sacrifice one of the considered axioms. Alternatively, one can also study more restricted models in which all the desired properties can be simultaneously satisfied. Indeed, we offer positive results in both directions, circumventing the NP-hardness results in several cases.
% (a) for the cases of sacrificing efficiency or any one of the considered axioms, and (b) for more restricted and structured frameworks. 
Finally, we believe that our work is making a pioneering contribution to the Computational Social Choice literature, by incorporating concepts and techniques from temporal graph theory, which is a novel approach in the field.
%, to the best of our knowledge.
% we note that our work is the first in the computational social choice literature, upon our knowledge, that has been based on ideas and techniques from temporal graph theory.
% in order to formulate and analyze the 
% temporal dimension of our Liquid Democracy model.
% concept of time-conscious.

\subsection{Related Work} 
% As it will be extensively discussed later, the model we suggest is a generalization of the one examined in \cite{brill2022liquid}. 
We discuss first some works related to Liquid Democracy, in a way that ``begins at the beginning'', following the suggestion made to the White Rabbit in ``Alice's Adventures in Wonderland''. The author of that novel, Charles Dodgson (also known by his pen name Lewis Carroll), 
%is renowned for the voting method named after him \cite{hemaspaandra1997exact}. Apart from that (and certainly unbeknown to him), 
as early as 1884 \cite{carroll1884principles}, considers an idea that was meant to be of vital importance for what we call Liquid Democracy today. According to \cite{origins}, it seems that he is the one that, before all else, discussed the aspect of giving the agents the power to transfer to others their acquired votes. On the other hand, it was Gordon Tullock \cite{tullock} who initiated the discussion about models that aspire to occupy the ground between direct and representative democracy, by suggesting a model that allows voters to decide whether they are interested in casting a ballot or delegate to another voter. Shortly after, unlike Tullock's suggestion, James Miller \cite{miller}, brought forward the idea that voters should not only choose their mode of participation but should also enjoy the ability to retract a previously given delegation in a day-to-day basis. At what concerns the nomenclature of LD, the precise origins are unknown. The best one could refer to, is its seemingly first \cite{paulin14} recorded appearance (in an obsolete wiki, preserved only on the Internet Archive \cite{archive1,archive2}), in which a user nicknamed ``sayke'' discoursed about a voting system that lies between direct and representative democracy and aims at increasing civic engagement. However, none of these sources discussed explicitly the aspect of transitivity of votes, as Dodgson did. Reinventions, amendments and compositions of these ideas started to appear at the early 00's and we refer to \cite{ford2014} for an overview of them. The earliest published works that incorporate the aspects of LD, 
(roughly) as we consider it today are \cite{ford,green,boldi,cohensius}. Nowadays, Liquid Democracy is one of the most active research areas in Computational Social Choice  \cite{brill2018interactive,paulin2020overview}.

As already mentioned, the primary motivation of our work is due to \cite{Nitsche} and the framework we suggest is a generalization of the model in \cite{brill2022liquid}. Furthermore, our optimization objective coincides with the one in \cite{escoffier,markakis2021approval}. To our knowledge, our work is the first that incorporates temporal aspects in LD models. 
% In the upcoming sections, we will explicitly reference works that are most relevant to our model, objective, or techniques. 
%Acknowledging that this list is not intended to be comprehensive, 
Many different models and questions related to Liquid Democracy have been examined. 
Indicatively, recently published works explored aspects including, the study of voting power concentration through the lens of parameterized complexity \cite{dey2021parameterized}, the efficiency of altering delegations to achieve consistency in participatory budgeting settings \cite{jain}, the application of power indices and criticality analysis to voters \cite{colley2}, and the evaluation of LD's epistemic performance \cite{revel}.
% we indicatively refer to  \cite{revel,jain,colley2,dey2021parameterized}, as well as to the surveys .

%%%%%%%%%%%%%%%%%%%%%%%%%%%%%%%%%%%%%%%%%%%%%%%%%%%%%%%%%%%%%%%%%%%%%%%%

%%%%%%%%%%%%%%%%%%%%%%%%%%%%%%%%%%%%%%%%%%%%%%%%%%%%%%%%%%%%%%%%%%%%%%%%

%%%%%%%%%%%%%%%%%%%%%%%%%%%%%%%%%%%%%%%%%%%%%%%%%%%%%%%%%%%%%%%%%%%%%%%%

\section{Temporal Liquid Democracy Elections} 
\label{sec:prelims1}
We consider elections in which a set $V$ of $n$ voters should reach a decision on a certain issue.
% \footnote{For the problems we study, it is irrelevant whether we have a single or multiple issues and hence we highlight that all of our results hold also for multi-issue election.} 
Apart from voting themselves, the participants are given two additional options: abstaining or delegating to other voters. The voters also have some time available to consider what to do (e.g. to get informed on the issue at hand or to observe other voters' choices) and they are allowed to change their mind, perhaps multiple times, until the actual election-day. We say that such an election is a \textit{Temporal Liquid Democracy Election}, a t-LD election in short, if it consists of two phases: 
\begin{itemize}
    \item A \textit{deliberation phase} of $L$ rounds, where at every time-instant $t\in [L]$, each voter $v$, has to choose whether to personally vote or not. If she decides to cast a ballot, we consider this as her final decision that will not change in the remaining time-steps. As long as a voter $v$ has not decided to cast a ballot herself till time $t$, she is asked to specify the following:
    \begin{itemize}
        \item[-] A set of approved voters $S_v^t\subseteq V\setminus\{v\}$ (which may be the empty set, if $v$ wants to abstain at round $t$), indicating the voters that she trusts to cast a ballot on her behalf, possibly with different levels of confidence. These voters may in turn also be willing to delegate their vote to other participants as well.
        \item[-] A (weak) preference ranking over the voters in $S_v^t$, which induces a partition of $S_v^t$ into preference groups, according to $v$.
        % , so that $rank_v^t(u)=i$ indicates that $u$ belongs to the $i$-th most preferred group of voters according to $v$, at time $t$, 
        This is accompanied by a positive integer score $sc_v^t(i)$, indicating the utility or happiness level that $v$ experiences if a voter from her $i$-th most preferred group at time $t$, will ultimately be selected as her immediate representative.
        \item[-] A non-negative integer-valued trust-horizon parameter $\delta_v^t$, with $\delta_v^t\leq t-1$, by which, she indicates approval for the views held by any voter in $S_v^t$ up to $\delta_v^t$ time-steps prior to time-instant $t$. 
    \end{itemize}    
    \item A \textit{casting phase}, in which all the voters that, during the deliberation phase, expressed willingness to vote (and only these), eventually cast a ballot on the issue under consideration. Every voter who did not previously declare an intention to vote, is being assigned a representative and a pre-specified delegation rule, that takes into account the entire deliberation phase, is being used to make such decisions. The winner(s) of the election are elected using a weighted voting rule, where the weight of a voter is determined based on the number of voters that she represents. 
\end{itemize}

For an illustrative exposition of our model we refer to the example provided at the end of the section. We now elaborate on the input that is required from the voters, during the deliberation phase. The preference ranking facilitates voters to express different levels of confidence towards other participants who could potentially represent them. Also, the scoring function allows the model to capture the cases where a voter is willing to either increase her scores over rounds due to becoming gradually more informed about another voters' opinions 
or in the opposite direction, decrease scores due to becoming more hesitant about who represents her. Realistically, we expect voters to have just a few preference groups, and hence they do not need to submit too many numerical parameters. Furthermore, the intuition behind the trust-horizon parameters is that the decision of a voter $v$ to approve a delegation to $u$ at time $t$, can be based only on looking at the behavior of $u$ in the previous rounds and up until time $t$ of the deliberation phase. Since a voter $v$ may not agree with $u$ in all previous time steps, the parameter $\delta_v^t$ specifies that $v$ agrees with the choices made by $u$ at any preceding time that is no more than $\delta_v^t$ time-instants before $t$.
%(but not necessarily at subsequent times, since the delegation may get retracted). 
A simple case to have in mind is when $\delta_v^t=t-1$ (i.e., $v$ trusts whatever $u$ has chosen at any time in the past). 
%or $\delta_v^t=0$ (i.e., $v$ trusts only the opinion of $u$ in the current round). 
If this property holds for every voter $v$ and for any $t\in [L]$, we say that the election profile is of \textit{retrospective trust}.
%if  such that $v$ prefers to delegate her ballot at time $t$, it holds that $\delta_v^t=t-1$, or in other words, when the time-horizon declared by each voter implies trust that goes arbitrarily far back in time. 
Finally, note that under the suggested model, the voters' declared sets, their preference orders, their scores and trust-horizon parameters may change arbitrarily between subsequent time-instants.

\smallskip \noindent \textbf{Variants of the Model and Practical Considerations.} 
Customization is key to the proposed decision-making model, which offers a range of possibilities to enhance its practicality and reproducibility. For instance, it would be more natural to allow a voter to specify a different time-horizon parameter for each of her approved representatives. Notably, our findings are not impacted by the assumption of a uniform trust-horizon for approved representatives. Furthermore, we have assumed that once a voter expresses the desire to cast a ballot, she will no longer change her opinion until the election-day. This assumption is justified by the fact that once a voter has committed to becoming more informed on the topic, participating in further deliberation is deemed redundant. Nonetheless, the assumption is made only for technical convenience and could be dropped. Moreover, although our aim is to examine the model in its fullest generality, we stress that in potential real-life implementations, the voters may not need to submit all the information that we have described in every round. In particular, the scoring function could be automatically generated by the system, given the (weak) ranking on $S_v^t$ submitted by each voter. E.g., one  could use the Borda-scoring function (as in \cite{brill2022liquid}), under which, at any time-instant, a voter assigns a score of 1 to her last preference group, a score of $2$ to her second to last group, etc) or any other appropriate method. We highlight that our model is a strict generalization of the model considered in \cite{brill2022liquid}, not only because of the temporal dimension but also because of the more general scoring functions that we allow. The trust-horizon parameter could also be pre-specified, so that the voters do not need to submit any information regarding it, either by assuming that the trust of every voter goes arbitrarily back in time or for a fixed number of steps prior to each approval. Finally, if voters have the same preferences for consecutive time-steps, they would not need to re-specify them.

\smallskip \noindent \textbf{Delegation Rules.} In the elections we consider, we essentially have three types of participants. We refer to the voters that declared intention to vote as \textit{casting voters}, and these will be the only voters who will indeed finally cast a ballot at the election-day. 
Furthermore, the non-casting voters that will abstain from the election are precisely those who do not approve anyone at the final time-step, e.g. a voter $v$ such that
% $\cup_{i\in[\lambda]}R_v^i(L)=\emptyset$
$S_v^L=\emptyset$. We refer to such voters as \textit{abstaining voters}. Finally, the rest of the voters will be called \textit{delegating voters}. 
As evident from Section \ref{sec:intro}, and as will be further illustrated by the example at the end of this section, the temporal dimension could be considered valuable when the examination of the isolated instance at $t=L$ cannot produce a feasible solution (i.e. delegation cycles or paths towards abstainers are unavoidable) or its feasible solutions are not good enough. A delegation rule is a mechanism that 'resolves delegations' and addresses such problematic cases, or in other words, a procedure 
%for breaking cycles and avoiding paths towards abstainers, that could have occurred at the election-day. Such a mechanism, 
that ultimately assigns to each delegating voter, a casting voter, possibly via following some path of trust relationships. More formally, a delegation rule is a function that takes as input the voters' preferences, 
as reported during the entire deliberation phase of a t-LD election, and outputs a path to a casting voter,
%representative 
for every delegating voter. 
%by following the representatives suggested for every voter, one can deduce a casting voter for every delegating voter. 
A valid delegation rule should ask casting voters to vote, abstaining voters to abstain and should not suggest any delegation path towards an abstainer or introduce delegation cycles.

\smallskip \noindent \textbf{Temporal Graphs.} The driving force in our work is to model and analyze t-LD elections using principles from temporal graph theory. We start with a basic overview of the concept and the terminology of temporal graphs and following this, we will introduce some notation that we will use in the remainder. In high level, a temporal graph is nothing more than a simple, called \textit{static}, graph in which a temporal dimension is being added, i.e., a graph that may change over time. Frequently, a temporal (multi)graph is being expressed as a time-based sequence of static graphs.
% , denoted by $G(V,E,L):=(G_t(V,E_t))_{t\in \{1,2,\dots,L\}}$. 
For convenience, we will use an equivalent definition, under which, a (directed) temporal (multi)graph $G(V,E,\tau,L)$ is determined by a set of vertices $V$, a (multi)set of directed, temporal edges $E$, a discrete time-labelling function $\tau$ that maps every edge of $E$ to a subinterval of $[1,L]$, and a lifespan $L \in \mathbb{N}$. If the edges of $E$ are weighted according to a function $w:E\rightarrow \mathbb{N}$, then we say that $G$ is weighted. The interval $\tau(e)$, for an edge $e$, indicates that $e$ is available at the time-instants that belong to $\tau(e)$. We say that each edge $e$ is assigned an interval labeling $\tau(e)=[s_e,t_e]$, 
% $\subseteq [1,L]\cap \mathbb{N}$
 (possibly, $s_e=t_e$ if the edge is available for a single time-instant) and by allowing $G$ to be a multigraph\footnote{We are using multigraphs instead of (simple) graphs merely for technical convenience, and we note that, alternatively, one could work with graphs by letting $\tau$ be a function that maps edges to a set of subintervals of $[1,L]$.
 % and $w$ be a function with domain $E\times [L]$\gp{is the domain correct?}.
 } it is permitted for an edge to be present in multiple (disjoint) time-intervals. Unless otherwise stated, henceforth, by the term \textit{graph}, we denote a weighted directed temporal multigraph. For more details on temporal graphs we refer to a relevant survey \cite{michail2016introduction}, as well as to the fundamental and influential works \cite{kempe2000connectivity,kostakos2009temporal,mertzios2013temporal}.
% Similarly, a path is \textit{time-conscious} if instead of the relation (\ref{backwardTC}) it holds that $1\leq t_1\leq t_2 \leq \dots \leq t_k \leq L$
% (usually called 'time-respecting' or 'journey' or simply 'temporal' in the related literature\gp{not sure if we have used the same type of quotation marks in the entire paper. I'll check it again}). 
The \textit{static variant} of a temporal graph is the static graph that emerges if we ignore the time-labels of its edges. 
% (i.e.,the graph $(V,\cup_{t\in\{1,2,\dots,L\}}E_t)$)
% We call a graph \textit{temporal directed tree} if it is a temporal graph whose static and undirected variant is a tree.
% We say that a graph is \textit{rooted at a vertex $r$} if there exists a (directed) path towards $r$ from every other vertex. 
We call a graph \textit{temporal directed tree rooted at vertex $r$} if its static variant contains a directed path towards $r$ from every other vertex and its undirected variant is a tree. 
A crucial concept for our work, in the context of temporal graphs, is the notion of time-conscious paths, that satisfy a monotonicity property regarding the temporal dimension of their edges. Consider a temporal graph $G(V,E,\tau,L)$, 
% a vertex set $V'\subseteq V$ and a vector $\delta$ that, for every $v\in V'$, contains 
coupled with a tuple $\delta_v=(\delta_v^t)_{t\in [L]}\in \mathbb{N}^{[L]}$ for every vertex $v$ of $V$. Let also $\delta = (\delta_v)_{v\in V}$. We say that a path in $G$ from $v_1$ to $v_{k+1}$ is $\delta$-\textit{time-conscious} 
% \footnote{Time-conscious paths should not be confused with the notions of 'time-respecting' or 'temporal' path, frequently used in the corresponding literature. This distinction will be crucial for proving Theorem 1.} 
if it can be expressed as an alternating sequence of vertices and temporal edges  
$(v_i,(e_i, t_i),v_{i+1})_{i\in [k]}$, 
% $(v_1,(e_1, t_1), v_2,(e_2, t_2),\dots , v_k,(e_k, t_k), v_{k+1})$, 
such that for every $i\in [k]$ it holds that $e_i =
(v_i, v_{i+1})\in E$, $t_i \in \tau(e_i)$ and for every $i\in [k-1]$ it holds that $t_i\geq t_{i+1} \geq t_i-\delta_{v_i}^{t_i}$. Similar notions have been applied to various domains including convenient flight connections detection \cite{wu2016efficient}, information diffusion \cite{huang2015minimum} and  infectious disease control through contact tracing \cite{temporal-walks}. In the remainder of Section \ref{sec:prelims1}, it will become more clear how this notion fits in our framework. We also call $\delta$-time-conscious, a temporal directed tree, rooted at a vertex $r$, if all its paths towards $r$ are $\delta$-time-conscious.
% and that $v_i\in V'$ and $v_{k},v_{k+1}\in V$.
% Given an integer value $\delta$, we say that a path in an interval graph from $v_1$ to $v_{k+1}$  is \vm{$\delta$-}\textit{time-conscious}\footnote{Time-conscious paths should not be confused with the notions of 'time-respecting' or 'temporal' path, frequently used in the corresponding literature. This distinction will be crucial for proving Theorem 1.} if it can be expressed as an alternating sequence of vertices and temporal edges  
% $(v_i,(e_i, t_i),v_{i+1})_{i\in [k]}$, 
% % $(v_1,(e_1, t_1), v_2,(e_2, t_2),\dots , v_k,(e_k, t_k), v_{k+1})$, 
% such that for every $i\in [k]$ it holds that $e_i =
% (v_i, v_{i+1})$, $t_i \in \tau(e_i)$ and $
% t_i\geq t_{i+1} \geq t_i-\delta.$ 
% If in a graph $G$ rooted at a vertex $r$, every path to $r$ is $\delta$-time-conscious, we say that $G$ is a $\delta$-\textit{time-conscious graph}. 
Finally, we conclude by noting that illustrative examples of some of the terminology discussed here, can be found at Appendix \ref{sec:appendix examples}.

\smallskip \noindent \textbf{Modelling t-LD Elections as Temporal Graphs.} The deliberation phase of a t-LD election can be modeled as a weighted directed temporal multigraph $G(V\cup\{\triangledown\},E,\tau,L,w,\delta)$ that is formed by
\begin{itemize}
    \item a vertex in $V$ for every voter of the electorate, as well as a special vertex $\triangledown$, connected only with the casting voters,
    \item a multiset $E$ of temporal edges that represent the approvals for delegation or ballot casting per round via a function $\tau$ that assigns a time-label to every edge,
    \item a lifespan $L$ that represents the duration of the deliberation phase,
    \item a function $w$ that assigns a weight to every edge $(v,u)$ of $E$, according to $sc_v^t$, provided that $t\in \tau((v,u))$,
        \item a vector $\delta$ that, for every voter $v$, contains a tuple $(\delta_v^t)_{t\in [L]}$, as declared by $v$ during the deliberation phase. For convenience, we allow $\delta$ to have some empty entries, corresponding to casting voters or to time steps during which the corresponding voter abstained.  
\end{itemize}

\noindent We note that if a casting voter had indicated preferences for potential representatives before deciding to cast a ballot, these preferences, and their corresponding edges, can be safely disregarded. More precisely, only the following two types of edges may exist: directed edges of the form $e=(v,u)$ for $v \in V\setminus C$ and $u\in V$ with $\tau(e)=[s_e,t_e]$, indicating that at any time-instant $t\in [s_e,t_e]$, voter $u$ belongs to $S_v^t$,
% $\cup_{i\in [\lambda]}R_u^i(t)$, 
and directed edges $e=(v,\triangledown)$ for $v \in C$ with $\tau(e)=[s_e,L]$, indicating that from time $s_e$ and onwards, voter $v$ agrees to cast a ballot. 
% If there is a time-instant $t$ and a vertex $u$ such that there does not exist any out-going edge from $u$ at time $t$, then $u$ is an abstainer, at that time. 
Furthermore, we will proceed by assuming that the set of voters $V$ is implicitly partitioned into three sets, as has been explained before: the set of casting voters $C$, the set of abstaining voters $A$ and the set of delegating voters $D$. More formally, $C=\{v\in V: (v,\triangledown)\in E\}$, 
% $A$ is the set of vertices which do not have an out-going edge at time $L$,
$A=\{v\in V\setminus C: L\notin \tau((v,u)), \text{for any } (v,u) \in E\}$ 
and $D=V\setminus (C\cup A)$. 
% The weight function indicates the weak preference order as well as the scores submitted by a delegating voter at the deliberation phase. More specifically, the weight of an edge $(u,v)$ that is present at time $t\in \tau(e)$ expresses the order of $v$ in the preference relation of $u$ at time $t$. Formally, if $u\in V$ and $d(u)$ is the out-degree of $u$, then for the weight function $w:E\rightarrow\mathbb{N}_{\geq 1}$ it holds that there exists a value $r(u) \in \{1,2,\dots, d(u)\}$ such that $\{w(e):e\in d(u)\} = \{ 1,2, \dots, r(u)\}$. If $w((u,v_1))\leq w((u,v_2))$ this is interpreted as a delegation to $v_2$ is at least as good as a delegation to $v_1$ for voter $u$. We finally set the out-degree of $\triangledown$ to zero.  
The weight function $w$ indicates the cardinal preferences of a voter, as implied by the scores that accompany her preference rankings during the deliberation phase. Additionally, 
% Formally, if there is a pair $(i,t)$ such that $u \in R_v^i(t)$ then $w(v,u) = sc_v^i(t)$. 
for convenience, we set to zero the weights of edges $(v,u)$ such that $v$ corresponds to a casting or an abstaining voter. 
%In other words we define $w(v,u)=0, \forall (v,u) \in E\cap(C\times (V\cup \{\triangledown\})) \cup (A\times V))$. 
% This choice can be justified by the optimization objective, discussed in the paragraph "Electorate's Satisfaction".
This choice can be justified by the upcoming discussion of the optimization objective in the ``Electorate's Satisfaction'' paragraph.
%that follows. 
Given a graph $G(V\cup\{\triangledown\},E,\tau,L,w,\delta)$ that models a t-LD election, a delegation rule returns, for every delegating voter $v$, a weighted directed temporal path from $v$ to $\triangledown$. Such a path infers an assignment of every delegating voter to a casting one. 
% weighted directed temporal graph $T$, that is a, rooted at $\triangledown$, subgraph of $G$, the vertex set of which is a superset of $V\setminus A$.
% should be a function $f: G \rightarrow T(V',E',\tau,L,w,\delta)$ where $V'\supseteq V\cup\{\triangledown\} \setminus A$, $E'\subseteq E$ and $T$ is a weighted directed temporal graph rooted at $\triangledown$. 
% Given a graph $G(V\cup\{\triangledown\},E,\tau,L,w,\delta)$ that models a t-LD election, a delegation rule should be a function $f: G \rightarrow T(V',E',\tau,L,w,\delta)$ where $V'\supseteq V\cup\{\triangledown\} \setminus A$, $E'\subseteq E$ and $T$ is a weighted directed temporal graph rooted at $\triangledown$. 
% that spans the vertices of $V\setminus A$. 
% Therefore, as expected, $T$ induces, for every delegating voter, a unique temporal path to $\triangledown$, and hence to a casting voter. 
% \gp{could we avoid the assumption?} We assume that there is a $\delta$-time-conscious path from every delegating voter to $\triangledown$. 
A delegation rule is called \textit{efficient} if its output can be computed in polynomial time in the input size.
% Additionally, we assume that the degree of any vertex $v$ such that $(v,\triangledown)\in E$ is exactly $1$. 
% For any other vertex $u\in V$, we assume that it holds that $\bigcap_{e=(u,v)}\tau(e)=\emptyset$, i.e.,that no such voter is indifferent at any time unit, and that every voter, at every time unit, accepts only a single other voter as her delegate. Note that these assumptions could be relaxed or removed.

\smallskip \noindent \textbf{Axiomatic Principles.} We discuss here the axioms that constitute the main focus of our work, namely time-consciousness and confluence. We have extensively discussed these two properties in the analogous paragraphs of Section \ref{sec:intro} and we now formally define them using the framework of temporal graphs. 
% \begin{definition}
Given a graph $G(V\cup \{\triangledown\},E,\tau,L,w,\delta)$ that models a t-LD election, a delegation rule is
\begin{enumerate}
    \item \textit{time-conscious}, if for every delegating voter $v$, the delegation path output for $v$ is a $\delta$-time-conscious directed temporal path,
    % s, rooted at vertex $\triangledown$, the vertex set of which is a superset of $V\setminus A$, 
    \item \textit{confluent}, if the union of the paths output for all the delegating voters is a directed temporal tree, rooted at vertex $\triangledown$, that spans the vertices of $V\setminus A$.
\end{enumerate}
% \end{definition}

%Both definitions are designed to ensure that every delegating voter is assigned to a casting voter, through a path of trust-relationships. 
\noindent The definition of time-consciousness guarantees that all paths suggested by the delegation rule satisfy the constraints imposed by the voters, regarding their trust-horizon parameters.
Hence, for any edge $(v, u)$ in an output path, $u$ must choose an action (edge) that she had declared at a time that was approved by $v$. The definition of confluence guarantees that for every delegating voter $v$, there is a unique path to a casting voter, that is intended to serve both $v$ and all voters who delegated to $v$. 
%In other words, the out-degree of every vertex in the union of the paths output by the rule should be at most $1$. 
In Appendix \ref{sec:appendix examples} we provide examples of time-conscious and confluent solutions, for further illustration. 

\smallskip \noindent \textbf{Electorate's Satisfaction.}
We make the usual assumptions for Liquid Democracy models that (a) voters completely trust their representatives and (b) trust between voters is transitive. This implies that if voter $v$ accepts voter $u$ as her potential representative, she concurs with any subsequent choice made by $u$ and also extends trust to any voter $w$ who may be entrusted by $u$. Hence, we note that the utility experienced by a delegating voter from a delegation rule can be considered as a local one, being contingent solely on the voter's immediate representative and not influenced by further choices made by the chosen representative. Hence the utility of a delegating voter, under a delegation rule, can be determined by the score that she declared for her immediate representative, specified by the rule.
% , if a delegation rule $f$ suggests the voter $v \in A_u^t,$ (for some $t\in [L]$) as the representative of $u$. More precisely, we assume the utility is given by 
% \begin{equation}
%     \label{eq:score}
%     val_f(u)=|A_u^t|+1-rank_u^t(v).
% \end{equation} 
Note that two different time-instants $t,t'$ may exist such that $u \in S_v^t \cap S_v^{t'}$. In these cases, given that the output of a delegation rule is a set of temporal paths, if the rule suggests a delegation from $v$ to $u$, it also explicitly specifies the time-instant at which the delegation will occur, say e.g. at time $t'$ and, thereby the utility of $v$ is equal to $sc_v^{t'}(i)$, if $u$ belongs to the $i$-th most preferred group of $v$, at time $t'$. Regarding now the casting voters, we do not take into account their utility since their will to cast a ballot has been realized; we do the same for abstaining voters. We consider as infeasible every solution that asks a casting (resp. abstaining) voter to delegate her ballot or abstain (resp. vote), and hence, our focus will be on the welfare of the delegating voters.
Finally, the quality of a rule is assessed by the total satisfaction it elicits from the electorate which is expressed as the sum of utilities of all delegating voters.
% , i.e. $\sum_{u}val_f(u)$, where the summation is over each delegating voter $u$ \vm{wasn't happy that we set the casting voters utility to 0, better to say that we do not care about them.}.
Our optimization objective then is to maximize the electorate's satisfaction, as defined by the \textsc{resolve-delegation} problem below.
%, which we will aim to solve in both a time-conscious and confluent manner. 

% In the quest for a delegation rule $f$ that maximizes the satisfaction of the electorate over all time-conscious and confluent delegation rules, we will study the \textsc{resolve-delegation} problem. More precisely we will focus on solving it in a time-conscious and confluent manner. The definition of the problem follows: 

% \begin{table}
% 	\centering
% 	\begin{tabular}{lp{6.4cm}}  
% 		\toprule
% 	 \multicolumn{2}{c}{\textsc{resolve-delegation}} \\
% 		\midrule
% \textbf{Instance:} & A graph $G(V\cup\{\triangledown\},E,\tau,L,w,\delta)$ that represents the deliberation phase of a t-LD election.
% % over $L$ rounds of an election on the voters of $V$, whose approvals and scores concerning any possible representative per round are given by $E$ and $\tau,w$.
% \\
% % \textbf{Output:} & A time-conscious directed temporal tree rooted at $\triangledown$ that spans the vertices of $V\setminus A$ and that maximizes the total utility of the voters, in polynomial time.\\
% \textbf{Goal:} & Compute a graph $T=\texttt{argmax}\{\sum_{e\in E'}w(e): T(V',E',\tau,L,w,\delta), V \cup\{\triangledown\}\setminus A \subseteq V' \subseteq V , E'\subseteq E\}$.\gp{in the non-confluent case, an extra information is needed! We could add a relevant footnote in the paragraph that precedes the Electorate's satisfaction.}\\
% % and are time-conscious directed temporal trees rooted at $\triangledown$
% 		\bottomrule
% 	\end{tabular}
% \end{table}

\begin{table}
	\centering
	\begin{tabular}{lp{6.4cm}}  
		\toprule
	 \multicolumn{2}{c}{\textsc{resolve-delegation}} \\
		\midrule
\textbf{Instance:} & A graph $G(V\cup\{\triangledown\},E,\tau,L,w,\delta)$ that represents the deliberation phase of a t-LD election.
% over $L$ rounds of an election on the voters of $V$, whose approvals and scores concerning any possible representative per round are given by $E$ and $\tau,w$.
\\
% \textbf{Output:} & A time-conscious directed temporal tree rooted at $\triangledown$ that spans the vertices of $V\setminus A$ and that maximizes the total utility of the voters, in polynomial time.\\
\textbf{Goal:} & Compute a weighted directed temporal path from each delegating voter to $\triangledown$, with the aim of maximizing the total utility derived from the delegating voters, defined as the sum of the weights of the paths' first edges.\\
% $\bigcup_{v\in D}P_v = G'(V',E',\tau,L,w,\delta)
% =\texttt{argmax}\{\sum_{e\in E'}w(e): V \cup\{\triangledown\}\setminus A \subseteq V' \subseteq V , E'\subseteq E\}$. \gp{so as to maximize the total utility derived by the delegating voters, i.e. sum of weights of first edge in each $P_v$}\\
% and are time-conscious directed temporal trees rooted at $\triangledown$
		\bottomrule
	\end{tabular}
\end{table}

\smallskip \noindent \textbf{Example.} As an illustration, consider the following instance of a t-LD election with 5 rounds and 6 voters, namely Alice, Bob, Charlie, Daisy, Elsa, and Fred. Their preferences are outlined below:
\begin{itemize}
\item Alice initially intended to delegate to Charlie. In the second round, she decided to get informed about the considered issue and vote.
% cast a ballot.
\item Bob did not participate in the deliberation phase during the first round, but approved Alice in the second round. In the third round, he revoked his approval of Alice and instead approved Chris and Elsa. Bob's approval of Elsa remained until the final round.
\item Charlie approved Alice only in the beginning of the election. He also approved Bob in the first and third round, but removed his approval (and abstained) in the second round. In the fourth round, Charlie approved both Daisy and Fred, but he removed his approval of Daisy in the final round.
\item Daisy expressed interest in being a casting voter from the beginning until the end of the deliberation phase.
\item Although Elsa intended to delegate her ballot to Fred at certain times, ultimately both refrained from participating in the election.
\end{itemize}

% \begin{itemize}
%     \item 
%     Even if during the first round Alice was intended to delegate to Charlie, at the second round, she decided to get informed about the considered issue and cast a ballot. 
%      \item Bob did not take part to the deliberation phase during the first round but he declared his approval of Alice during the second round. 
%      % More specifically, $S^1_B(2)=\{C\}$, $R^2_B(2)=\{A\}$ and $sc^i_B(2)=3-i$.    
%      During the third round, he revoked his approval towards Alice and instead approved Elsa. Only the approval towards Elsa lasted after the third round, till the final round.
%     \item Charlie approved Alice only in the beginning of the election and approved Bob in the first and third round, but he decided to remove this approval (and abstain) in the second round. At round $4$ he approved both Daisy and Fred but he removed his approval to Daisy at round $5$.
%     % In the third round, he participated again by submitting the following preferences: $R^1_C(3)=\{D\}$, $R^2_C(3)=\{B\}$ and $sc^1_C(3)=3, sc^2_C(3)=1$.
%     \item Daisy expressed interest in being a casting voter from the beginning until the end of the deliberation phase.
%     \item Elsa and Fred did not participate in the election and abstained. 
% \end{itemize} 

% \begin{wrapfigure}{r}{0.22\textwidth}
% \centering
% \fbox{\includegraphics[scale=0.8
% ]
% {example}}
% \caption{An example of the elections considered.}
% \label{fig:example}
% \end{wrapfigure}

        \setlength{\columnsep}{6pt}
\begin{wrapfigure}[17]{r}{-4cm}
        \raisebox{0pt}[\dimexpr\height-1.7\baselineskip\relax]{\includegraphics[scale=0.375]{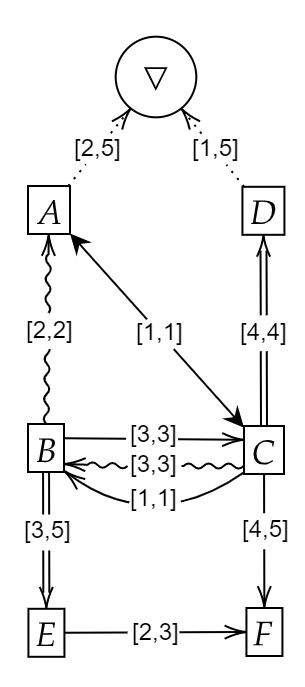}}%
\end{wrapfigure}
% The described instance can be represented by the graph of the side figure. Suppose also that for every $t\in\{1,2,\dots,5\}$, it holds that $\delta_B^t=2$ and that $\delta_C^t=t-1$. Also, say that the scores assigned by the voters to their approved representatives are depicted in the figure as follows: the curly edges are of weight $1$, the straight edges are of weight $2$ and the double-lined edges are of weight $3$. The remaining, dotted edges indicate the casting voters. The labels of the edges indicate the time-instants of their availability. Note that edge $(A,C)$ could be deleted since definitely Alice will cast a ballot. Observe that Alice and Daisy will form the set of casting voters. Additionally, Elsa and Fred will abstain. Furthermore, Bob will not delegate to Charlie at time $3$ since no $\delta$-time-conscious path from $C$ to $\triangledown$ begins at any time no later than round $3$ that would not disregard Bob's specified time-horizon. Similarly Charlie will not delegate to Bob at time $1$. Given also that Bob should not delegate to an abstaining voter, it turns out that in any feasible solution he should delegate to Alice at time $2$, even if this is not his favorite option. Then, there are two possible outcomes of a delegation rule, depending on the choices concerning Charlie. The edge that maximizes his utility is $(C,D)$ and therefore, the outcome of the optimal delegation rule that is both confluent and time-conscious is $\{(C,D),(D,\triangledown),(B,A),(A,\triangledown)\}$, with a total electorate's satisfaction score of $4$. 

\noindent The described instance can be visualized using the graph shown in the side figure. We assume that $\delta_B^t=1$ and $\delta_C^t=t-1$ for every $t\in \{1,2,\dots,5\}$. The scores assigned by the voters to their approved representatives are encoded by the form of the edges, where curly edges have weight $1$, straight edges have weight $2$, and double-lined edges have weight $3$. Dotted edges indicate the casting voters. The labels of the edges represent the time-intervals of their presence. In this instance, Alice and Daisy form the set of casting voters, while Elsa and Fred abstain. Therefore, edge $(A,C)$ can be removed, since Alice will definitely cast a ballot. In a $\delta$-time conscious solution, Bob would not delegate to Charlie, since no $\delta$-time-conscious path to $\triangledown$ 
using the edge $(B, C)$ exists, for instance, edge $(C,A)$ violates the time-horizon declared by Bob.
%at any time no later than round $3$ that would not disregard Bob's specified time-horizon. 
Similarly, Charlie would not delegate neither to Alice nor to Bob at time $1$. Since we do not allow Bob to delegate to an abstainer, he must delegate to Alice, whom she trusted at time $2$. Then, there are two possible outcomes for the delegation rule, depending on the choice made for Charlie.
The edge that maximizes Charlie's utility is $(C,D)$. Therefore, the optimal delegation rule that is both time-conscious and confluent, would suggest the set of paths $\{((C,D),(D,\triangledown)),((B,A),(A,\triangledown))\}$, achieving a total satisfaction score of $4$. Finally, in this example it is plainly evident how the temporal dimension comes to the rescue: if one were to focus solely on the snapshot taken at time $5$, disregarding the information garnered from the deliberation phase, the only option would be to ask Bob and Charlie to delegate to abstaining voters. 
%This would result in lost ballots and we regard this as an infeasible solution. 
Instead, our framework utilizes the information obtained throughout the deliberation phase to propose an outcome that avoids paths towards abstainers and delegating cycles.

% Concluding the section, we note that all missing proofs can be found in the Supplementary Material of our work.

%%%%%%%%%%%%%%%%%%%%%%%%%%%%%%%%%%%%%%%%%%%%%%%%%%%%%%%%%%%%%%%%%%%%%%%%
  
% \section{Οn Efficient Utility Maximizing Delegation Rules Satisfying Both Axioms}
\section{Computational Complexity of Resolving Delegations in t-LD elections}
\label{sec:hardness}
%%%%%%%%%%%%%%%%%%%%%%%%%%%%%%%%%%%%%%%%%%%%%%%%%%%%%%%%%%%%%%%%%%%%%%%%

In this section we explore the compatibility of the axioms we have put forward from Sections 1 and 2, with efficient computation. We highlight that all the missing proofs from the paper can be found in the Supplementary Material of our work. Unfortunately, our first result shows that it is impossible to have polynomially computable utility maximizing delegation rules that satisfy simultaneously the axioms of time-consciousness and confluence, unless P=NP, even under simple and natural restrictions. Before stating the result, we discuss the types of instances for which we establish hardness. To begin with, it is expected that in real-life elections, voters tend to exhibit a relatively stable and consistent opinion over time, and do not revise their preferences numerous times during the deliberation phase, due to the effort it would require to gather and process new information. Similarly, it is reasonable to expect that due to limited cognitive capacity, the voters are only able to partition their accepted representatives into a few disjoint preference groups. The theorem that follows demonstrates that the computational intractability of \textsc{resolve-delegation} persists even when we limit the voters to changing their minds at most once during the deliberation phase and partitioning their accepted representatives into at most two groups, at each round. Furthermore, it holds even for instances of retrospective trust, and with Borda-scoring functions. Therefore, the primary takeaway is that incorporating temporal aspects in conjunction with natural requirements does come at a computational cost.

\begin{thm}
\label{thm:np-hard}
\textsc{resolve-delegation} in a time-conscious and confluent manner is NP-hard, even for profiles of retrospective trust and under the Borda-scoring function.
\end{thm}

\begin{proof}
% We will show that given a weighted directed temporal graph $G(V\cup \{\triangledown\},E,\tau,L,w)$ that represents a deliberation phase of a t-LD election, it is NP-hard to find a delegation rule that is both time-conscious (TC) and confluent (C), proving that these axioms are incompatible with the requirement of efficiency (E), under the $P\neq NP$ computational complexity assumption. 
% More specifically, we will prove that the problem is NP-hard even for the most natural, Borda, scoring function. 

We provide here a description of the reduction and defer the proof to Appendix \ref{app:proofs}. Given a graph $G(V\cup \{\triangledown\},E,\tau,L,w,\delta)$ and a parameter $k$, we call $\Pi$ the decision variant of \textsc{resolve-delegation} in a time-conscious and confluent manner, which asks for the existence of a solution with total satisfaction at least $k$. 
% Let's call this problem \textsc{similarresolve-delegation[tc+c])}. 
% Observe that if $\Pi$ is NP-hard then \textsc{resolve-delegation[tc+c+e]} cannot be solved, since the efficiency requirement would contradict with the NP-hardness result. 
At what follows, we provide a reduction to $\Pi$ from the NP-hard problem \cite{huang2015minimum} \textsc{minimum temporal spanning tree} (t-\textsc{mst}), which we formally define shortly. Before moving on to the definition of t-\textsc{mst}, we note that in temporal graph theory the term \textit{time-respecting} is being used to describe, a temporal path $(v_i,(e_i, t_i),v_{i+1})_{i\in [\ell]}$, 
such that for every $i\in [\ell]$, it holds that $e_i =
(v_i, v_{i+1})$, $t_i \in \tau(e_i)$, and $1\leq t_1\leq t_2 \leq \dots \leq t_{\ell} \leq L$ (also called ``journey'' or simply ``temporal'' in the related literature). We also refer to Appendix \ref{sec:appendix examples} for an example. The difference between time-respecting and $\delta$-time-conscious paths is that  the paths of the former type are formed by edges whose time-stamps are in non-decreasing order of visiting, in contrast to the paths of the latter type, whose edges have time-stamps in non-increasing order, and on top of that, satisfy a waiting-time constraint indicated by vector $\delta$.

In the t-\textsc{mst} problem, we are given a temporal graph $G'(V',E',\tau',L',w')$, as well as a root vertex $u_0'\in V'$ and an integer $k'$. We are asked for a directed temporal tree of $G'$, called $T'$, of edge set $E''$, that spans the vertices of $V'$ and that has a time-respecting path from $u_0'$ to every vertex of $V'$, such that $\sum_{e\in E''}w'(e)\leq k'$. Note that t-\textsc{mst} is NP-hard even for the case where $w'(e)\in \{1,2\}, \forall e\in E'$, and for every $v\in V'$ there exists a $u\in V'$, such that $L'\in \tau'((u,v))$. It is without loss of generality to assume that $u_0'$ has no in-coming edges in $E'$. Furthermore, the hardness holds for instances in which for any pair of vertices $u,v$ of the input graph $G'(V',E',\tau',L',w')$, either $(u,v)\notin E'$, or there are two copies, $e_1$ and $e_2$ of $(u,v)$ in the multiset $E'$. In the second case, it also holds that $\tau'(e_1)=[1,L'-1], \tau'(e_1)=[L',L']$ and that $w'(e_1)=2, w'(e_2)=1$.

% an edge $(u,v)$ does not exist for any time-instant $t$, or such an edge exists with a weight of $2$ for all but the last time-instant, and with a weight of $1$ for the last time-instant.

% Minimum Temporal Spanning Tree is a problem that at first sight looks similar to $\Pi$. To highlight the main differences between the two, we point out that Minimum Temporal Spanning Tree is a minimization problem that searches for time-respecting paths on graphs where the paths that are being sought should have a common starting point in contrast with thetime-conscious paths with the same ending point of $\Pi$. Additionally note that in any instance of $\Pi$, if there is an edge $(u,\triangledown)\in E$ that is present at time $t$, then there does not exist a vertex $v \in V$ such that an edge $(u,v)$ exists at any time $t'\geq t.$ The reduction that follows handles these differences.
Given such an instance $(G'(V',E',\tau',L',w'),u_0',k')$ of t-\textsc{mst} we create an instance $(G(V\cup \{\triangledown\},E,\tau,L,w,\delta),k)$ of $\Pi$ as follows:

\begin{itemize}
    \item let $L=L'$,
    \item for every vertex $u' \in V'$ we add a vertex $u \in V$,
    \item for every directed edge $(u',v')\in E'$ we add a directed edge $(v,u)$ such that $w(v,u)=3-w'(u',v')$ (recall that $w'(u',v')\in\{1, 2\}$) and $\tau((v,u))=\tau'((u',v'))$,
    \item we add a special vertex $\triangledown$ and a directed edge $e=(u_0,\triangledown)$ such that $w(e)=0$ and $\tau(e)=[1,L]$,  
    \item we add one more special vertex $a\in V$,
    \item for every $t\in [L]$ and $v\in V$ such that there exists in $E$ an out-going edge from $v$ at time $t$ of weight $2$ but not of weight $1$, we add an edge (called ``dummy'') $e=(v,a)$ such that $w(e)=1$ and $\tau(e)=[t,t]$,
    \item for every vertex $v\in V\setminus \{a\}$ that corresponds to a non-casting voter and for every $t\in [L]$, we set $\delta_v^t=t-1,$ 
    \item we set $k$ to be $3(n-1)-k'$.
    % so as to make the election that is being specified by $G$, a t-LD election of retrospective trust \vm{find it better at this point to just say we set $\delta^t = t-1$.}.\gp{the problem is that $\delta$ is only defined for delegating voters so we cant say that $\delta^t_v = t-1,\forall v$. also current definition of 'retrospective trust' profiles may not be technically correct. lets discuss it. Perhaps $\delta^t_v = t-1,\forall v\in V\setminus \{a,u_0\}$}
\end{itemize}
For the remainder of the proof, we refer to Appendix \ref{app:proofs}. 
\end{proof}

We will now explore roads to circumvent the impossibility result of Theorem \ref{thm:np-hard}.
% propose some escape routes from the negative result of Theorem \ref{thm:np-hard}. 
Our proposal is to relinquish either the necessity for efficiency or one of the axioms of time-consciousness and confluence, in hopes of solving \textsc{resolve-delegation}. Our findings show that this strategy proves successful for some of the problems that emerge, which highlights that
% Our second suggestion, as outlined in Subsection \ref{sec:restrict}, is to focus on certain restricted settings of t-LD elections, for which \textsc{resolve-delegation[tc+c+e]} is solvable, as a means of showing that 
Theorem \ref{thm:np-hard}
is not devastating. Notably, most of the suggested procedures are simple enough and therefore are confirmed as strong contenders for practical applications.

    We begin with studying the easiest variant of \textsc{resolve-delegation} in which the requirement of time-consciousness is being disregarded. 
    This is mainly done for the sake of completeness since studying it requires overlooking the temporal dimension of the instance, which is the defining characteristic of our work. In order to solve efficiently \textsc{resolve-delegation} in a confluent but not necessarily time-conscious manner, the delegation rule can treat any input submitted by a voter at any time as if it was not subject to time-related constraints.  
    % we view any input submitted by a voter in any time-instant, as equally considered \gp{don't like the phrasing here, due to the scoring function, the inputs are not equally considered but could be seen as time-irrelevant}by the delegation rule.  
% For that, one can employ any established confluent delegation rule (provided that it can be adapted to fit our model), as long as it runs in polynomial time and maximizes the electorate's utility. 
Since confluence implies that the output should be a directed tree, and since the utility of each delegating voter is determined by its outgoing edge, then all edges of the tree with non-zero weight will contribute exactly once to the total satisfaction, and therefore, the objective is to find a (static) directed tree of maximum total weight, that is rooted at $\triangledown$ and spans the vertices of $V\setminus A$.
% directed spanning tree of maximum weight (after ignoring the abstaining voters), rooted at $\triangledown$, 
% that has a directed path towards the root from every vertex. 
To solve this problem we leverage the well-known algorithm by Edmonds \cite{edmonds1967optimum} (also independently discovered in \cite{bock1971algorithm,chu1965shortest} and improved in \cite{tarjan1977finding}) for the directed analog of the classic \textsc{minimum spanning tree} problem.\footnote{To be noted that in \cite{brill2022liquid}, a confluent delegation rule, referred to as MinSum, has been proposed, under a more restricted voting framework compared to ours, and its polynomial time computability has been very recently established \cite{colley2022unravelling}, using an approach that is also based on Edmonds' algorithm.} In this problem, given a weighted directed static graph $G(V,E,w)$ and a designated vertex $r\in V,$ we are asked for a subgraph $T$ of $G$, the undirected variant of which is a tree, of minimum total cost, such that every vertex of $G$ is reachable from $r$ by a directed path in $T$. It is important to note that in our case, the paths we need to compute are towards a fixed vertex, rather than originate from it. To apply Edmonds' algorithm, we adjust graph $G$
%$G(V\cup\{\triangledown\},E,\tau,L,w,\delta)$ 
to an appropriate graph $G'$,
% and run Edmonds' algorithm with input $(G',\triangledown)$, 
as indicated by Procedure \ref{proc1}.

\begin{thm}
     Procedure \ref{proc1} solves \textsc{resolve-delegation} in a confluent manner, in polynomial time. 
\end{thm}

\begin{algorithm}[h!]
\caption{A confluent and efficient utility maximizing delegation rule for input $G(V\cup\{\triangledown\},E,\tau,L,w,\delta)$.}\label{proc1}
\flushleft
{\tiny{1.}} $G:=$ static variant of $G$\\
% {\tiny{2.}} \textbf{For} every vertex $v\in C$: \vm{haven't we removed these edges?}\\ 
% {\tiny{3.}}\hspace{0.4cm} remove from $E$ all out-going edges of $v$ except $(v,\triangledown)$ \\
% {\tiny{4.}} \hspace{0.4cm} set the weight of $(v,\triangledown)$ to $n^2$ \gp{incorrect value. max-score*n}\\
{\tiny{2.}} $V':=V \cup \{\triangledown\} \setminus A $\\
{\tiny{3.}} $E':=\{(u,v):(v,u) \in E \wedge v,u\in V' \}$\\
{\tiny{4.}} \textbf{For} every edge $e'\in E'$:\\ 
{\tiny{5.}}\hspace{0.4cm} $w'(e'):=-w(e),$ where $e'\in E'$ corresponds to $e\in E$\\
{\tiny{6.}} remove duplicates from $E'$, retaining only the min-weight edge\\
{\tiny{7.}} let $G'$ be the (static) directed weighted graph $(V',E',w')$\\
% {\tiny{8.}} run Edmonds' algorithm with input ($G',\triangledown)$\\
% {\tiny{9.}} \textbf{Return} the edges of $E$ \vm{change to:the path from each delegating voter to $\triangledown$} that correspond to the result of Step 10
{\tiny{8.}} $T':=$ outcome of Edmonds' algorithm with input ($G',\triangledown)$\\
{\tiny{9.}} \textbf{Return} the path from each $v\in D$ to $\triangledown$ inferred by $T'$
\end{algorithm}

We now shift our focus to efficient utility maximizing delegation rules that satisfy time-consciousness but are not necessarily confluent. 
%Delegation rules that fall short of being confluent, may require a voter to distribute her assigned ballots between multiple different voters. 
Despite not necessarily resulting in a tree structure, such a rule should still suggest a precise path to a casting voter for every delegating voter $v$. Then, the utility of $v$ will be derived from her immediate representative (i.e., the weight of the first edge) in that path, regardless of whether other paths going through $v$ may exist for serving other voters who have delegated to $v$. The question of why non-confluent delegation rules merit investigation is discussed in \cite{brill2022liquid}. It was discovered that, among a large family of delegation rules, only non-confluent rules possessed the potential to satisfy the axiom of {\it copy-robustness}, an axiom that is also motivated by practical considerations \cite{behrens}.
% , but its study is beyond the scope of our work. 
Moreover, there are non-confluent rules with desirable properties that have been previously studied, such as the Depth-First-Delegation rule that precludes the possibility of Pareto-dominated delegations \cite{kotsialou2018incentivising}. Hence, it is not unprecedented to sacrifice confluence on the altar of attaining other desirable attributes. However, quite surprisingly, even in the absence of a requirement for a confluent rule, \textsc{resolve-delegation} remains NP-hard, as shown by the following theorem. Notably, the result holds even for simple scenarios that involve only a brief deliberation phase, uniform trust-horizon parameters across all voters and a lone delegating voter and it is orthogonal to the result of Theorem \ref{thm:np-hard} since it explicitly uses the fact that the considered elections are not of retrospective trust.

% The main reason why this holds is the difference in the computational complexity of the reachability problem in static versus in temporal graphs. \gp{reachability is polynomial in temporal graphs, if we don't have the $\delta$ constraint. avoid the last sentence? however it is interesting that the hardness holds because even fundamentally simple graph problems become hard in our setting}

% However, once again, and quite surprisingly, it is impossible to solve \textsc{resolve-delegation} in a time-conscious manner as the following theorem indicates. The main reason why this holds is the difference in the computational complexity of the reachability problem in static versus in temporal graphs.

\begin{thm}
\label{2nd-hardness}\textsc{resolve-delegation} in a time-conscious manner is NP-hard, even for profiles with only a single delegating voter. 
\end{thm}

Continuing with our study of efficient utility-maximizing delegation rules that are time-conscious but not necessarily confluent, we now turn to exploring potential workarounds to the impossibility result of Theorem \ref{2nd-hardness}. To overcome the computational intractability, we restrict ourselves to the still hard variant where the voters share the same time-horizon parameter and propose the following relaxations:
\begin{itemize}
    \item[(a)] Assuming retrospective trust profiles, i.e. $\delta_v^t=t-1$, for every voter $v$ and for every time-step $t$ such that $v$ approves to delegate her ballot at time $t$. These profiles are motivated by the fact that in real life, we do not expect voters to change their opinion in an arbitrary manner, and hence it is likely that a delegating voter trusts another voter for all the previous time instants before $t$. 
    % We note that the NP-hardness reduction of Theorem \ref{2nd-hardness} does not use such instances. 
\item[(b)] Permitting walks instead of only paths, or in other words allowing for revisits to vertices, along a path from a delegating voter to a casting one. This enlarges the solution space and can be helpful towards achieving time-consciousness in certain instances, as it may be necessary to go through a cycle before being able to satisfy the time constraints. For an illustration, we refer to 
% the relevant example in 
Appendix \ref{sec:appendix examples}. 
\end{itemize}

The approach of neglecting confluence, enables the development of local delegation rules, likewise the rules studied in \cite{caragiannis}, that make a decision for every voter completely independent of the choices
made for the rest of the electorate. For the two relaxations suggested in the previous discussion, we suggest a simple procedure that, in high level, visits every vertex $v$, corresponding to a delegating voter $v$, in a sequential manner, and for each such vertex, it detects a feasible, i.e. $\delta$-time-conscious, way to reach $\triangledown$, that uses the out-going edge of $v$ of maximum possible weight. The aforementioned way of reaching a casting voter can be computed by a suitable modification of the temporal analog of the Breadth-First search algorithm from \cite{mertzios2013temporal}, in the case where the input profile is of retrospective trust and by using the polynomial procedure that is based on Dijkstra's algorithm, from \cite{temporal-walks}, in the case where walks are allowed and all voters share the same trust-horizon parameter. 

Concerning the first relaxation, in \cite{mertzios2013temporal}, a polynomial-time algorithm was suggested to solve a (more general than what we need in our setting) problem, called \textsc{foremost path}. In this, we are given a (unweighted) directed temporal graph $G(V,E,\tau,L)$, a source vertex $v\in V$, a sink vertex $u\in V$, and a time-instant $t_{start}\in [L]$, and we are asked to compute\footnote{To be more precise, the goal is to select the path that minimizes the arrival time but for our purposes, this objective is superfluous (but harmless).} a time-respecting path from $v$ to $u$, that starts no sooner than $t_{start}$ (or report that such a path does not exist). Recall that, the definition of a time-respecting path has been provided in the proof of Theorem \ref{thm:np-hard}.
% This algorithm works roughly as follows: initially it considers $v$ as the only explored vertex and for every time $t\in\{t_{start},\dots,L\}$, it picks all vertices that have been already explored, one-by-one, and examines all their out-going edges that are present at time $t$. If an edge $e$ with $t\in \tau(e)$ reaches a vertex $u$ that has not been explored yet, then $e$ is selected as an edge of the desired path from $v$ to $u$ and $u$ is marked as explored at time $t$. 
% For our purposes, we will call \textsc{foremost path} the problem in which another vertex $u\in V$ is also given as input and we are only asked for a time-respecting path from $v$ to $u$, that starts no sooner than $t_{start}$ (or report that such a path does not exist). The modification of the algorithm from \cite{mertzios2013temporal} in order to fit to the new definition of the problem, is straightforward.

% At what follows, we will use the notion of $d$-time-conscious walks, in analogy to $d$-time-conscious paths. A
For the second relaxation, we give first the following definition: Given, a temporal graph $G(V,E,\tau,L,\delta)$ in which all entries of the vector $\delta$ coincide with a fixed value $\Delta$, a temporal walk $p$ of $G$ of length $\ell$, say $p=(v_{i-1},(e_i,t_i),v_i)_{i\in [\ell]}$ such that $v_i$'s are not necessarily all pairwise distinct,
% $p=\{v'_{0},(e_1',t_1),v'_1,(e_2',t_2),v'_2,\dots,v'_{k-1},(e_k',t_k)\}$
% \gp{what about using the following notation for paths: \\$p=((v'_{i-1}, (e'_i,t_i)))^k_{i=1}$}
is called $\Delta$-restless if for every $i\in [\ell]$ it holds that $e_i =
(v_{i-1,} v_i)$ and that $t_i \in \tau(e_i)$ and for $i\in [\ell-1]$ it holds that $t_{i} \leq t_{i+1} \leq
t_{i} + \Delta$. To solve efficiently the relaxation of \textsc{resolve-delegation} in a time-conscious manner, when walks are allowed, we will utilize the procedure from \cite{temporal-walks}, that outputs\footnote{Once again, the problem studied in \cite{temporal-walks} is more general than the problem we need to consider here, both in terms of the input graph and the optimization objective(s), but it can be easily adapted to meet our requirements.} a $\Delta$-restless temporal walk between two specified vertices, for any fixed parameter $\Delta$. For compactness, we provide a unified presentation of the positive results, under Procedure \ref{proc2}, which handles both relaxations. In the statement and analysis of this procedure, we will use the term \textit{journey} to refer either to a path when dealing with the first relaxation or to a walk when discussing the second relaxation.

\begin{thm}
Procedure \ref{proc2} solves \textsc{resolve-delegation} in a time-conscious manner, in polynomial time, for profiles of retrospective trust. Moreover, the same holds for the variant of the problem where walks are allowed, for profiles in which there is a common, fixed trust-horizon parameter, for all voters and all time-steps.
\end{thm}

% \vspace{-1cm}
\begin{algorithm}[h!]
\caption{A time-conscious and efficient utility maximizing delegation rule for input $G(V\cup\{\triangledown\},E,\tau,L,w,\delta)$ applicable for profiles of retrospective trust or profiles in which walks are allowed.}\label{proc2}
\flushleft
% {\tiny{1.}} delete vertices of $A$ and let $V:=V\setminus A$\\
{\tiny{1.}} \textbf{For} every edge $e\in E$:\\
{\tiny{2.}}\hspace{0.3cm} replace $\tau(e)$ with $\{L+1-t: t\in \tau(e)\}$\\
{\tiny{3.}} \textbf{For} every vertex $v\in V\setminus A$:\\
{\tiny{4.}}\hspace{0.3cm} $E_v:=\{(v,u)\in E:u\in (V\cup\{\triangledown\})\}$\\
{\tiny{5.}}\hspace{0.3cm} \textbf{While} a journey from $v$ to $\triangledown$ hasn't been found and $|E_v|> 0$:\\
{\tiny{6.}}\hspace{0.6cm} \textbf{If} $(v,\triangledown)\in E_v$: pick $(v,\triangledown)$ as the path from $v$ to $\triangledown$, and exit \\
\hspace{0.82cm}the while loop\\
{\tiny{7.}}\hspace{0.6cm} $\tilde{e} := \texttt{argmax}\{w(e): e\in E_v\}$   \\
% {\tiny{9.}}\hspace{0.6cm} $E_v:=E_v\setminus\{e\}$\\
{\tiny{8.}}\hspace{0.6cm} \textbf{If} walks are allowed:\\
{\tiny{9.}}\hspace{0.9cm} $G':=(V\cup\{\triangledown\},(E\setminus E_v)\cup \{\tilde{e}\},\tau,L,\delta)$\\
{\tiny{10.}}\hspace{0.9cm} search for a $\Delta$-restless walk from $v$ to $\triangledown$ in $G'$, if it doesn't\\\hspace{1.2cm}exist, remove $\tilde{e}$ from further consideration\\
{\tiny{11.}}\hspace{0.6cm} \textbf{Else}:\\
{\tiny{12.}}\hspace{0.9cm} $G':= (V\cup\{\triangledown\},(E\setminus E_v)\cup \{\tilde{e}\},\tau,L)$ \\
{\tiny{13.}}\hspace{0.9cm} $t_{start}:=s_{\tilde{e}},$ where $\tau({\tilde{e}})=[s_{\tilde{e}},t_{\tilde{e}}]$ \\
{\tiny{14.}}\hspace{0.9cm} solve \textsc{foremost path} for $(G', v,\triangledown, t_{start})$, if a solution\\\hspace{1.2cm}is not found, remove $\tilde{e}$ from further consideration\\
{\tiny{15.}} \textbf{Return} the set of determined journeys 
\end{algorithm}

We conclude with studying the problem \textsc{resolve-delegation} in a time-conscious and confluent manner, but now without the requirement of computational efficiency. Clearly, if polynomial solvability is no longer a worry, a straightforward brute-force procedure, that in time exponential in the number of edges and in $L$ examines all possible trees, can be utilized to maximize the voters' satisfaction. However, our objective goes beyond this. First, we aim at developing a procedure that could be well-suited for scenarios where the deliberation phase is prolonged, being exponentially dependent in only one of its input parameters. Additionally, observing that the most definitive parameter of \textsc{resolve-delegation} is the number of delegating voters $|D|$ (upper bounded by $n$), we focus on designing an algorithm with a running time exponentially dependent only on $|D|$, which would be suitable for practical use in any relatively small community. Yet, this is not possible without further assumptions, given the negative result of Theorem \ref{2nd-hardness}, that holds even for a single delegating voter. As before, we resort to instances of t-LD elections of retrospective trust and obtain the following result.

\begin{thm}
    \textsc{resolve-delegation} in a time-conscious and confluent manner is solvable in time exponential in $|D|$ and polynomial in the remaining input parameters, for profiles of retrospective trust. 
\end{thm}

\section{Conclusions}

Succinctly speaking, the main attributes of Liquid Democracy are the (i) voters' ability to cast a ballot, (ii) ability to delegate voting rights, (iii) transitivity of delegations, (iv) ability for topic-specific delegations, (v) ability to modify or recall a delegation. Our work is the first, upon our knowledge, in the Computational Social Choice literature, that studies a model that satisfies every each of the above features. 
Motivated by the suggestion of \cite{Nitsche}, on the addition of a temporal dimension in the algorithmic considerations of LD models, and building upon \cite{brill2022liquid}, 
we studied a LD framework from a viewpoint that lies in the middle ground between algorithmic and axiomatic approaches. We intentionally gave significant emphasis on developing a general model for incorporating temporal aspects and we feel it opens up the way for several promising avenues for future research. The first is to examine whether time-consciousness (or other time-related axioms) is compatible with established axioms for LD frameworks. Also an intriguing topic is to identify further realistic families of instances for which all properties studied here can be simultaneously satisfied. It would also be interesting to check whether the positive results we present still apply for further generalizations of t-LD elections, e.g. when voters are able to use a more powerful language
to express complex preferences. The delegating voters' preferences over the final representatives and the casting voters' preferences over the issue at hand, as well as an egalitarian objective or restrictions on the maximum in-degree or on the maximum path-length in the output of a delegation rule, deserve further examination.
% , while the egalitarian objective is also worth considering. 
We finally suggest experimental or empirical evaluations of LD frameworks that take into account temporal considerations.
\newpage
\bibliography{ecai}

\clearpage
\setcounter{page}{1}

\appendix

% \noindent{\underline{\Large Technical Appendix \hspace{2.5cm} paper ID: 1997}}
\noindent{\underline{\Large Technical Appendix \hspace{5.2cm}}}

\section{Missing Proofs From Section 3}
\label{app:proofs}

\subsection{Proof of Theorem 1}
The special cases for which the hardness holds, stated in the statement of the theorem, simply follow by the construction. We make the following observations regarding the deliberation phase of the t-LD elections represented by $G$.

% that the created instance of \textsc{decision(resolve-delegation[tc+c])} fulfils all structural requirements that are needed in order to express a deliberation phase of a t-LD election:
\begin{itemize}
    \item[-] The only vertex that has an out-going edge to $\triangledown$ is $u_0$. Such an edge is available at the final time-instant $L$ and thus the voter that corresponds to $u_0$ agrees to cast a ballot till the end of the election; which makes her the only casting voter. 
    % Furthermore, $u_0$ does not approve any other voter as her potential representative at any preceding time.
        \item[-] The only vertex that doesn't have an out-going edge at time $L$ 
        % and consequently that it corresponds to an abstaining voter 
        is $a$. More precisely, $a$ has no out-going edges during $[1,L]$ and thus the corresponding voter opts to abstain from the beginning until the end of the election; which makes her the only abstaining voter.
        \item[-] The weights of the out-going edges of every vertex $v$ of $V\setminus \{\triangledown,u_0\}$ at any time-instant $t$, indeed express a weak ranking over the voters that are being approved by $v$ at time $t$, due to the dummy edges. 
        % if an edge $(v,u)$ is available at time $t$, either $v$ has only out-going edges of unit weight or has both out-going edges of weight 1 and 2, due to 
        % the dummy edges.
        % \gp{Note that edges of the second type always appear at subsequent time-instants than edges of the first type. Do we need this property?}
        % Thus, at such a time $t$, the corresponding voter either only (equally) approves all voters that belong to a certain set or she indicates two non-empty sets $X,Y
        % \subseteq V\setminus\{v,\triangledown\}$ such that $v$ prefers all voters in $X$ more than voters in $Y$ but she also accepts the voters from $Y$. 
\end{itemize}

   Before continuing, we observe that vertex $a$ as well as the edges towards $a$ do not affect the rest of the reduction since such edges do not belong to any path to $\triangledown$, and are not part of any feasible solution of $\Pi$. Hence, it is safe to focus on the subgraph of $G$ induced by $V\cup \{\triangledown\} \setminus \{a\},$ which, for simplicity, will be called $G(V\cup\{\triangledown\},E,\tau,L,w,\delta)$, from now on.

For the forward direction, say that $(G',u_0',k')$ is a \texttt{YES}-instance of t-\textsc{mst} having $T'$ as a certificate. We will prove that $(G,k)$ is also a \texttt{YES}-instance of $\Pi$.  We select an arbitrary path $p'$ of $T'$ that has $u_0'$ as its source vertex, and we rename its vertices and edges so as $p' = (u_{i-1}',(e_{i}',t_{i}),u_{i}')_{i \in [q]},$  for some $q\in \{1,2,\dots,n\}$.
% $$ p' = \{u_0',(e_1',t_1),u_1',(e_2',t_2),u_2',\dots,u_{q-1}',(e_q',t_q),u_q'\}.$$ 
Since $T'$ is a subgraph of $G'$, the existence in $p'$ of the edge $e_i'=(u_{i-1}',u_i')$, for $i\in [q]$ and for which $t_i \in \tau(e_i')$ implies the existence of an edge $e_i=(u_i,u_{i-1})$ in $G$ that is present at time $t_i$. 
Combining these edges we prove the existence of a path $p=(u_i,(e_i,t_i),u_{i-1})_{i\in [q]}$, in $G$.
% $p$ of $G$ that is
% $$p=\{u_q,(e_q,t_q),u_{q-1},(e_{q-1},t_{q-1}),u_{q-2},\dots,u_{1},(e_1,t_1),u_0\}.$$
Since $p'$ is time-respecting, it holds that $1\leq t_{i-1}\leq t_i\leq L, i\in \{2,3,\dots q\},$
and hence, given that the elections represented by $G$ are of retrospective trust, $p'$ is a $\delta$-time-conscious path. Finally, the path $p\cup (((u_0,\triangledown),1),\triangledown)$ is a $\delta$-time-conscious path from $u_q$ to $\triangledown$. Combining such paths for every vertex $u_q\in V\setminus\{u_0\}$, we can create a subgraph $T$ of $G$ that is a $\delta$-time-conscious tree rooted at $\triangledown$ that spans the vertices of $V\cup\{\triangledown\}$. 

We now focus on the cost of the edges in $T$. The cost of all edges of $T'$ is a sum of values $w'(e_i')\in \{1,2\}$ and the number of edges in $T'$ are exactly $n-1$, where $n=|V'|=|V|$. Lets call $d_{T'}$ the number of edges of weight $2$ in $T'$ then $n-1+d_{T'}\leq k'.$ Furthermore, it holds that for every edge of weight $2$ (resp. 1) in $T'$ there is an edge of weight $1$ (resp. 2) in $T$ and vice versa. Given that $w(u_0,\triangledown)=0,$ the total weight of edges of $T$ is
\begin{align*}
n-1+(n-1-d_{T'})=2(n-1)+(n-1)-k' \geq k.
\end{align*} 

For the reverse direction suppose that there is a directed temporal tree $T$ that verifies a \texttt{YES}-solution of $\Pi$ and say that $E_T$ is its edge set. Since $T$ is rooted at $\triangledown$ and the only edge incident to it is $e=(u_0,\triangledown),$ then $e$ is definitely part of $T$. Consider the graph $T'$ that corresponds to the subgraph of $G'$ that contains $e$ as well as an edge $(u',v')$ if and only if $(v,u)$ belongs to $E_T\setminus\{e\}$. The fact that $T'$ is a time-respecting directed temporal tree that spans the vertices of $G'$ and has a path from $u_0'$ to every vertex of $G'$, follows by similar arguments to the forward direction of the proof. 

We now need to prove that the total weight of the edges of $T'$ is at most $k'$. It is known that the total weight of the edges of $T$ is at least $k=3(n-1)-k'$. Lets call $d_T$ the number of edges of weight $2$ in $T$, then $(n-1)+d_T\geq 3(n-1)-k'.$ Since every edge of weight $2$ (resp. $1$) of $G'$ corresponds to an edge of weight $1$ (resp. $2$) in $G$ and vice versa, the total weight of the edges in $T'$ equals $$(n-1)+(n-1-d_T) \leq 2(n-1)+(n-1-3(n-1)+k') = k',$$ and this concludes the NP-hardness proof.\hfill$\qed$
% and, as a consequence, the incompatibility of the considered axioms.

\subsection{Proof of Theorem 2}
It is immediate that the complexity of Procedure \ref{proc1} is polynomial in the input size. The fact that the output of Procedure \ref{proc1} is a directed temporal tree that spans the vertices of $(V\cup \{\triangledown\}) \setminus A$ follows from the fact that Edmonds' algorithm returns a tree that spans the set of vertices of $G'$, which equals $(V\cup \{\triangledown\})\setminus A$. Additionally, the result of the run of Edmonds' algorithm is a graph in which every vertex of $G'$ is reachable from $\triangledown$. Given that every edge of $G'$ has a corresponding edge in $G$ of reverse orientation, it holds that the output of Procedure \ref{proc1} contains a path to $\triangledown$ from every vertex of $V\setminus A$. Note also that Procedure \ref{proc1} is a valid delegation rule since it assures that abstaining voters will not be asked to vote or to delegate their rights, that delegating voters will not be asked to delegate to an abstaining one and that casting voters will be asked to cast a ballot. The proof of these claims is straightforward.
% The requirements concerning abstaining voters simply follow from Step 4, whereas the requirement about casting voter follows from the fact that any vertex $v\in C$ is left with only one out-going edge (the one to $\triangledown$) by Step 3.

We finally need to argue about the optimality of the algorithm concerning the maximization of electorate's satisfaction. This follows by the fact that any solution by Edmonds for $(G',\triangledown)$ corresponds to a feasible solution of \textsc{resolve-delegation}, with the same cost, and vice versa.\hfill$\qed$
% Suppose that $T_1$
% ((V\cup \{\triangledown\}) \setminus A,E'(T),w)$ 
% is the subtree of $G$ returned by Procedure \ref{proc1}, having $E_1$ as its edge-set, and that there is another directed subtree of $G$, rooted at $\triangledown$, called $T_2$, with $E_2$ being its edge-set,
% (V \cup \{\triangledown\} \setminus A,E'(T'),w)$, 
% such that $\sum_{e \in E_1}w(e)<     \sum_{e \in E_2}w(e)$. But then, by the construction, it should also hold that there exists a pair of trees $T_1',T_2'$ that are subgraphs of $G'$ containing a path from $\triangledown$ to every of its vertices and are such that $\sum_{e \in E_1}w'(e)> \sum_{e \in E_2}w'(e)$. But this contradicts the optimality of the outcome of Edmonds' algorithm and concludes the proof.

\subsection{Proof of Theorem 3}

Given a graph $G(V\cup \{\triangledown\},E,\tau,L,w,\delta)$ that models the deliberation phase of a t-LD election and a parameter $k$, we call $\Pi$ the decision variant of the problem \textsc{resolve-delegation} in a time-conscious manner, which asks for the existence of a solution with total satisfaction 
%subgraph of $G$, say $T(\hat{V},\hat{E},\tau,L,w,\delta)$, that is formed by the union of $|D|$ $\delta$-time-conscious paths, each one starting from a vertex that corresponds to a delegating voter in $D$ and ending at $\triangledown$, for which it holds that the sum of the weights of the paths' first edges are 
at least $k$.
% $\hat{V}\supseteq V$ and that 
% $\sum_{e\in \hat{E}}w(e)\geq k$. 
At what follows, we provide a reduction to $\Pi$, from a problem called \textsc{restless temporal path}, that was shown to be NP-hard\footnote{The NP-hardness result from \cite{casteigts2019computational} pertains to the undirected variant of the problem. Nevertheless, an analogous result can be shown for the directed variant by a straightforward transformation.}
% the directed version can also be shown to be NP-hard simply by replacing each undirected edge $(u',v')$ in the temporal graph with two directed edges $(u',v')$ and $(v',u')$ having the same time label as the original edge. \vm{reduce}} 
in \cite{casteigts2019computational}. 

In the \textsc{restless temporal path} we are given, a temporal graph $G'(V',E',\tau',L')$, two distinct vertices $x, y$ of $V'$ and an
integer $\Delta$. The question is to determine whether a $\Delta$-restless path from $x$ to $y$ exists in $G'$, where a temporal path $p=(v'_{i-1},(e_i',t_i),v'_i)_{i\in [\ell]}$
% $p=\{v'_{0},(e_1',t_1),v'_1,(e_2',t_2),v'_2,\dots,v'_{k-1},(e_k',t_k)\}$
% \gp{what about using the following notation for paths: \\$p=((v'_{i-1}, (e'_i,t_i)))^k_{i=1}$}
is called $\Delta$-restless if for all $i\in [\ell]$, it holds that $e_i' =
(v'_{i-1}, v'_{i})$, $t_i \in \tau'(e_i')$, and for all $i\in [\ell-1]$ we have $t_{i} \leq t_{i+1} \leq
t_{i} + \Delta$ (in case $t_1+\Delta > L$, the rightmost inequality is trivially satisfied); an example of such a path is provided in Appendix \ref{sec:appendix examples}. The hardness holds even for temporal graphs with a lifespan of $3$ time instants and even if $\Delta=1$.
% Furthermore, we define the problem \textsc{multiple restless temporal paths} that given a directed temporal graph $G$, a set of vertices $S$ and a special vertex $t$, asks for a directed $\Delta$-restless temporal path from every vertex of $S$ to $t$. This problem is also NP-hard since it generalizes   Restless Temporal Path. 
% We will now prove that it is NP-hard not only to resolve delegations that maximize the social welfare, but also to check whether a path from every delegating voter to a casting voter exists. 

Given an instance $(G'(V',E',\tau',L'),x,y,\Delta)$ of \textsc{restless temporal path} we create an instance $(G(V\cup\{\triangledown\},E,\tau,L,w,\delta),k)$ of $\Pi$ as follows: 
\begin{itemize}
\item let $L=L'+1$,
    \item let $G$ be the reversed graph of $G'$, i.e. $V=V'$ and for every edge $e'=(v,u)$ of $E'$ we add an edge $e=(u,v)$ in $E$ such that $\tau(e)=\tau'(e')$ and $w(e)=1$,
    % $E=\{(u,v):(v,u)\in E'\}$
    % the graph that has the same set of vertices and an edge $(u,v)$ for every edge $(v,u)$ of $G'$
    % and for every pair of corresponding edges $((u,v),(v,u))=(e,e') \in E\times E'$, we set $\tau(e)=\tau(e')$ and $w(e)=1$,
    % (with the same time labels),
    \item we add in the vertex set of $G$ a special vertex $\triangledown$ and in $E$ an edge $(x,\triangledown)$ with $w((x,\triangledown))=0$ and $\tau((x,\triangledown))=[1,L]$, 
    % and we set $w(e)=1, \forall e\in E\setminus\{e\}$, 
    \item for every edge $(v,u)$ and for every time-instant $t\in [L]$ such that $v\neq x$ and $t\in \tau((v,u))$, we set $\delta_v^t=\min\{t-1,\Delta\}$,
    % \item for every vertex $v\neq x$ 
    % that corresponds to a delegating voter \vm{but only $y$ is a delegating voter}\gp{We should say 'the voters who expressed willingness to vote at any time', right?}
    % and for every time-instant $t\in [L]$ for which there exists a vertex $u\in V$ such that $(v,u)\in E$ and $t\in \tau((v,u))$, we set $\delta_v^t=\min\{t-1,\Delta\}$,
    \item we add a dummy vertex $a$ with no out-going edges and an edge $(y,a)$ with $w((y,a))=1$ and $\tau((y,a))=[L,L]$,
    \item we set $k=1$.
    % we set all entries of $\delta$ to be equal to $\Delta$ and $k=1$.
\end{itemize}

We proceed with a few observations on the profile of t-LD elections that $G$ models. First, we have that at time $L$ all voters that correspond to vertices other than $x,y$ want to abstain, 
% i.e. $\{(u,v)\in E: u\notin \{x,y\}\wedge L\in \tau((u,v))\}=\emptyset$,
since there are no edges in $G'$ that are present at time $L'+1$. On the other hand, the only edge towards $\triangledown$ in $G$ has $x$ as its head, and therefore, $x$ corresponds to the only casting voter. Therefore, the only delegating voter in the created instance is $y$, since it corresponds to the only vertex that has an outgoing edge to a vertex other than $\triangledown$, at the final time-instant. Furthermore, since all approved representatives of $y$ are tied in the first place of her ranking, and given that the only edge towards $\triangledown$ is from $x$, asking for a time-conscious solution with a total utility of at least $1$ is equivalent to determining whether a time-conscious path from $y$ to $x$ exists.
% the maximization of the social welfare is equivalent to the problem of finding whether a time-conscious path from $y$ to $x$ exists.  

For the forward direction, say that $p$ is a directed $\Delta$-restless path from $x$ to $y$ that verifies that $(G',x,y)$ is a \texttt{YES}-instance of the \textsc{restless temporal path} problem. Then, the edges of $G$ that correspond to edges of $p$ in $G'$ induce a path from $y$ to $x$ in $G$, since they are of reverse orientation, and consequently a path from $y$ to $\triangledown$. To prove that the path, called $p'$, from $y$ to $\triangledown$ is $\delta$-time-conscious, we focus on an arbitrary pair of consecutive edges of $p$, namely $((u,v),t'),((v,z),t'')$. Given that $p$ is $\Delta$-restless in $G'$, it holds that $t'\leq t''$ and that $t'' \leq t'+\Delta$. But then, by the construction of $G$, the path $p'$ includes the following two edges: $((z,v),t''),((v,u),t')$, for which $t''\geq t' \geq t''-\Delta = t''-\delta^{t''}_z$. By applying the same argument for each pair of consecutive edges of $p$, we conclude that $p'$ is indeed $\delta$-time-conscious.

% $t_2\geq t_1$ and that $t_1\geq t_2-\delta$, which satisfy the condition $t_i\geq t_{i+1} \geq t_i-\delta$ and, therefore, the considered path of $G$ is indeed time-conscious.

For the reverse direction, say that $(G(V\cup\{\triangledown\},E,\tau,L,w,\delta),k)$ is a \texttt{YES}-instance of $\Pi$, or in other words that there exists 
% a delegation rule that suggests to all non-casting and non-abstaining voters, a casting voter that is reachable by a time-conscious path. Hence, there is 
a $\delta$-time-conscious path from $y$ to $\triangledown$. Given that the only edge towards $\triangledown$ is from $x$ and that it is present at any time-step of the deliberation phase, there also exists a $\delta$-time-conscious path, say $p$, from $y$ to $x$. In analogy to the forward direction, selecting an edge $(u,v)$ of $E'$ if and only if the edge $(v,u)$ of $E$ belongs to $p$, establishes a $\Delta$-restless path from $x$ to $y$ in $G'$, verifying that $(G'(V',E',\tau',L'),x,y)$ is a \texttt{YES}-instance of \textsc{restless temporal path}.\hfill$\qed$

\subsection{Proof of Theorem 4}   
It is immediate that the complexity of Procedure \ref{proc2} is polynomial in the input size and that the procedure returns a valid delegation rule as it does not ask abstaining voters to vote or delegate their rights, it does not ask a delegating voter to delegate to an abstainer, and it asks every casting voter to cast a ballot. Therefore, we only need to prove that the output is indeed a set of feasible $\delta$-time-conscious paths from all delegating voters to $\triangledown$, that maximizes the electorate's utility.

We begin by proving that the journeys
% , both referred to as journeys hereinafter, 
returned by Procedure \ref{proc2} indeed maximize the electorate's utility, assuming that the if and else blocks in lines 8 and 11 output a feasible solution, whenever there exists one. To that end, we first note that since it is sufficient to return a journey from a vertex $v$ to $\triangledown$ that is totally independent from the returned journey from every other vertex to $\triangledown$, the maximization of the utility of the electorate boils down to the problem of maximizing the utility of every individual voter. Consider now a delegating voter $v$.  
%and towards a contradiction, assume that there is a feasible $\delta$-time-conscious journey from $v$ to $\triangledown$ that satisfies her more than the output of Procedure \ref{proc2}. Hence, there should be a pair of feasible $\delta$-time-conscious journeys $p_1,p_2$ from $v$ to $\triangledown$ in $G$ such that if $p_1$ is the output of Procedure \ref{proc2} concerning vertex $v$ and if we call $e_i$ the first edge of $p_i$ for $i\in \{1,2\}$, then $w(e_1)< w(e_2)$. 
Since Procedure \ref{proc2} examines the out-going edges of $v$ in order of decreasing weight, and it stops at the very first time it finds a feasible solution for $v$, this is obviously a solution of maximum utility for $v$. 
%would have considered $e_2$ before $e_1$ and since a journey from $v$ to $\triangledown$ that uses $e_2$ exists, it would have selected it, which leads to a contradiction.

To prove feasibility of the outcome, starting from the case where the input profiles are of retrospective trust, we claim that the replacement of the edge time-labels ensures that a time-respecting path from $v$ to $\triangledown$ in $G'$ corresponds to a $\delta$-time-conscious path from $v$ to $\triangledown$ in $G$. Say that two consecutive edges $e_1'$ and $e_2'$ are being used in such a time-respecting path, at time-instants $t_1'$ and $t_2'$ respectively. Then, $t_1'\leq t_2'$. Equivalently, there are two consecutive edges in $G$, say $e_1$ and $e_2$ that are present at time-instants $t_1$ and $t_2$ respectively. Since $t_1'\leq t_2'$, it also holds that $L+1-t_1\leq L+1-t_2$ and therefore, $t_1\geq t_2$. Hence $e_1$ can be placed before $e_2$ in a $\delta$-time-conscious path in $G$, provided that the input is of retrospective trust. By the optimality of the outcome of the algorithm for solving \textsc{foremost path}, from \cite{mertzios2013temporal}, one can deduce a feasible time-respecting path from $v$ to $\triangledown$ in $G'$, which in turn implies the existence of a $\delta$-time-conscious path from $v$ to $\triangledown$ in $G$.
% To prove feasibility of the outcome, starting from the case where the input profiles are of retrospective trust, we claim that the replacement of the edge time-labels ensures that a time-respecting path from $v$ to $\triangledown$ in $G'$ corresponds to a $\delta$-time-conscious path from $v$ to $\triangledown$ in $G$ and vice versa \vm{do not need vice versa here. Need to start from a time-respecting path output by foremost path and deduce a time-conscious path}. Say that two consecutive edges $e_1$ and $e_2$ are being used in such a $\delta$-time-conscious path of $G$ at time-instants $t_1$ and $t_2$ respectively. Then, $t_2\leq t_1$. Equivalently, there are two consecutive edges in $G'$, say $e_1'$ and $e_2'$ that are present at time-instants $L+1-t_1=t_1'$ and $L+1-t_2=t_2',$ which comes after $t_1'$, given that $t_2\leq t_1$, respectively. Hence $e_1'$ can be placed before $e_2'$ in a time-respecting path in $G'$. By the optimality of the outcome of the algorithm for solving \textsc{foremost path}, from \cite{mertzios2013temporal}, one can deduce a feasible time-respecting path from $v$ to $\triangledown$ in $G'$, which in turn implies the existence of a $\delta$-time-conscious path from $v$ to $\triangledown$ in $G$. The reverse direction is analogous. 

Coming now to the case where walks are allowed, the feasibility follows by the algorithm suggested in \cite{temporal-walks} for determining whether a $\Delta$-restless walk between two specified vertices exist. We need again to use similar arguments to the ones used for the retrospective trust case, showing that the replacement of the edge time-labels in the first lines of Procedure \ref{proc2}, ensures that a $\Delta$-restless path from $v$ to $\triangledown$ in $G'$ corresponds to a $\delta$-time-conscious path from $v$ to $\triangledown$ in $G$.\hfill$\qed$

\subsection{Proof of Theorem 5} 
We suggest a procedure which consists of two components: a reduction from \textsc{resolve-delegation} in a time-conscious and confluent manner
to the \textsc{directed minimum steiner tree} (\textsc{d-mst})  problem and an execution of an appropriate algorithm for the latter. 
% The first part, is (significantly different but) inspired by \cite{axiotis2016size}. \gpred{what about deleting the last sentence?}
We begin with defining \textsc{d-mst}. In that, we are given a (static) directed edge-weighted graph $G'(V', E', w')$, where $w':E'\rightarrow \mathbb{N}$, a source $r'\in V'$, a set of vertices $\hat{V}\subseteq V'$ called terminals, and we are asked for a subgraph of $G'$ that includes a directed path from $r'$ to any (terminal) vertex of $\hat{V}$, of minimum possible total weight.
% such that $\sum_{e\in E''} w'(e)\leq k'$. 
% An algorithm that runs in time $3^{|\hat{V}|}|V'|^{O(1)}$ exists (which is a straightforward modification of the classical algorithm for the -undirected- \textsc{minimum steiner tree} problem \cite{dreyfus1971steiner}) and further slight improvements on the running time have also been suggested (see e.g. \cite{jones2013parameterized}). 

Given a graph $G(V\cup \{\triangledown\},E,\tau,L,w,\delta)$ that models the deliberation phase of a t-LD election of retrospective trust, we refer to $\Pi$ as the \textsc{resolve-delegation} problem in a time-conscious and confluent manner.
% , which asks for the existence of a solution with a total satisfaction of at least $k$.
We will present a reduction from $\Pi$ to \textsc{d-mst}. Consider an instance of $\Pi$, say $G(V\cup\{\triangledown\},E,\tau,L,w,\delta)$. We will now construct an instance $(G'(V',E',w'), r', \hat{V})$ of \textsc{d-mst}, and an example of the proposed construction can be found in Figure \ref{fig:reduction}.
\begin{itemize}
\item we add a source $r'$ in $V'$, that corresponds to $\triangledown$,
        \item we add a (terminal) vertex $v'$ in $\hat{V}$, for every vertex $v$ of $D$, a (non-terminal) vertex $v'$ in $V'\setminus \hat{V}$, for every vertex $v$ of $C$ and we call all such vertices ``special'',
    \item for every edge $e=(u,v) \in E$ such that $u\in D$ and $v\in (V\cup\{\triangledown\})\setminus A$ and for every $t \in \tau(e)$, we add a (non-terminal) vertex named $(e,t)$ in $V'\setminus \hat{V}$,
    \item for every pair of edges $e_1=(u,v),e_2=(v,z)$ of $E$, such that $u,v,z \in (V\cup \{\triangledown\})\setminus A$ and for every $t_1\in \tau(e_1)$ and $t_2\in \tau(e_2)$ with $t_1\geq t_2$, we add in $E'$ a directed edge from $(e_2,t_2)$ to $(e_1,t_1)$ (provided that these vertices exist) of weight $max(u)-w(e_1)+min(u)$, where $max(u)$ (resp. $min(u)$) is the maximum (resp. minimum) weight of out-going edges of $u$, available at any time-instant, 
    % if $e_1$ is the $i$-th preference of $u$, in chronological order (which equals $1$ if it is (chronologically) the last option for voter $u$,$2$ if it is (chronologically) the 2nd to last option of $u$, ... and $d$ if it is (chronologically) her first option). 
    \item for every edge $e=(v,\triangledown)\in E$ and for every $t\in \tau(e)$, we add in $E'$ a directed edge from $r'$ to $(e,t)$ of zero weight, 
    \item for every edge $e=(u,v)\in E$ such that $u \in D$ and $v\in (V\cup \{\triangledown\}) \setminus A$ and for every $t\in \tau(e)$, we add in $E'$ a directed edge from $(e,t)$ to $u'$ of zero weight.
    % \item we set $k'=\sum_{u\in D}(max(u)+min(u))-k$.
    \end{itemize}

\begin{figure}[h!]
\label{fig:reduction}
% \centering
\fbox{\includegraphics[scale=0.405
]
{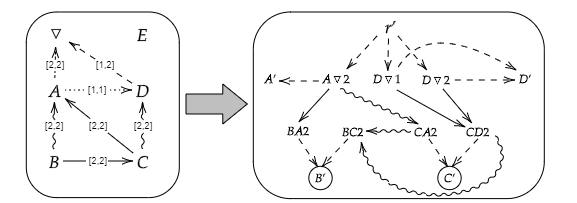}}
% \captionsetup{labelformat=empty,font={small,it}}
\centering \caption{The weighted directed temporal graph of an instance of problem $\Pi$ (left) that corresponds to a t-LD election with 1 abstaining voter $(E)$, 2 casting voters $(A,D)$ and 2 delegating voters $(B,C)$ and the weighted directed static graph of the corresponding instance of \textsc{d-mst} (right), where circled vertices indicate terminals. Dashed, curly and straight edges are of weights $0,1$ and $2$, respectively.}
\end{figure}

 At what follows we will prove that every feasible solution for an instance of $\Pi$, given by a graph $G$ of total utility at least $k$, implies the existence of a feasible solution for \textsc{d-mst} in $(G',r',\hat{V})$ of cost at most $k'$, and vice versa, where $k'=\sum_{u\in D}(max(u)+min(u))-k$. We start with the forward direction. We need to show that given a $\delta$-time-conscious temporal tree of $G$, called $T$, which spans the vertices of $V\setminus A$, is rooted at $\triangledown$ and its total weight is at least $k$, one can deduce a feasible solution $T'$ of weight at most $k'$ in the created instance of \textsc{d-mst}. 
Obviously $T$ induces a path from every vertex $v\in V\setminus A$ to $\triangledown$, which is of unit length if $v\in C$, and of length greater than $1$, otherwise. For such a path of unit length, there is also a directed path from $r'$ to $v'$ in $G'$, that is of zero cost, which is formed by the edges $(r',((v',\triangledown),$L$))$ and $(((v',\triangledown), $L$),v')$. Hence, the non-terminal vertex that corresponds to a casting voter, can be reached from $r'$ with no cost, using exactly 2 edges. Let all such pairs of edges belong to $T'$. 

We move on to temporal paths of length greater than 1 in $T$. We select an arbitrary path $p$ of $T$ having $\triangledown$ as its sink vertex and a vertex $u_0$, that corresponds to a delegating voter, as its source, and we rename its vertices and edges so as $u_q=\triangledown$ 
and $p=(u_{i-1},(e_i,t_i),u_i)_{i\in[q]},$ for some $q\in [n]$.
% $$ p = \{u_0,(e_1,t_1),u_1,(e_2,t_2),u_2,\dots,u_{q-1},(e_q,t_q),u_q=\triangledown\}.$$ 
Since $p$ is a $\delta$-time-conscious path, for $i\in [q]$, it holds that $e_i=(u_{i-1},u_i)$, $t_i\in \tau(e_i)$ and, furthermore, for $i\in [q-1]$ it holds that $t_i \geq t_{i+1}$. By the construction, there also exists a path $p'$ in $G'$ that can be expressed by the sequence of vertices $\{r',((u'_{q-1},u'_q),t_q),((u'_{q-2},u'_{q-1}),t_{q-1}),\ldots,((u'_1,u'_2),t_2),((u'_0,u'_1),t_1),u'_0\}$.
All in all we have constructed a directed path from $r'$ to $u'_0$, and $u'_0$ is a terminal vertex since $u_0$ corresponds to a delegating voter. By repeating the same procedure for every path $p$ of $T$ that has $\triangledown$ as its source, 
% and by adding all the corresponding paths $p'$ to $T'$, 
we can form a subgraph of $G'$, say $T'$, that consists of paths from $r'$ to every terminal vertex of $\hat{V}$. 

 It remains to be proven that $T'$ is a tree of cost no more than $k'$.
 Firstly, suppose that the undirected variant of $T'$ contains a cycle. By a similar reasoning to the above, we can prove that the undirected variant of $T$ would also contain a cycle, which is a contradiction. 
% The cost of the first and last edge of $p'$ is zero.
We will now prove that if the total weight of the edges in $T$ is at least $k$ then the weight of the edges in $T'$ is at most $k'$. Observe that for every vertex $v\in V\setminus A$, there is only one other vertex $u\in (V\cup \{\triangledown\}) \setminus A$, such that $(v,u)$ in $T$ (and similarly for $T'$) and say that, for convenience, the weight of an edge $(v,u)$ that belongs to $T$ is denoted by $w(v)=w(v,u)$. Then, the total weight of the edges in $T$ is
$\sum_{v\in D}w(v)\geq k.$ Let us now focus on a pair of paths of $T$ and $T'$, namely $p$ and $p'$  respectively. Say that $p=(u_{i-1},(e_i,t_i),u_i)_{i\in [q]}$ and that $p'$ is being formed by following the sequence of vertices $
% p = &\{u_0,(e_1,t_1),u_1,(e_2,t_2),u_2,\dots,u_{q-1},(e_q,t_q),u_q\}\\
\{((u'_{q-1},u'_q),t_q),((u'_{q-2},u'_{q-1}),t_{q-1}),\dots,((u'_1,u'_2),t_2),((u'_0,u'_1),t_1)\}$ in $G'$. The weight of all edges in $p$ equals $w(e_1)+w(e_2)+\dots+w(e_{q-1})+w(e_{q})= w(u_0)+w(u_1)+\dots+w(u_{q-2})+w(u_{q-1})$. On the other hand, the weight of all edges in $p'$ equals $max(u_{q-1})-w(u_{q-1})+min(u_{q-1})+max(u_{q-2})-w(u_{q-2})+min(u_{q-2})+\dots+max(u_0)-w(u_0)+min(u_0)$. Hence, for every such pair of paths, if $w_p$ is the total weight of the path $p$ (similarly for $p'$) and if $N_p$ is the set of non-sink vertices of path $p$, i.e. $N_p=\{u_0,u_1,\dots,u_{q-1}\}$, then $$w_{p'}=\sum_{u\in N_p}(max(u)+min(u))-w_p.$$

Before continuing, we note that any rooted directed (towards the root) tree with $\ell$ leaves can be divided into paths $P_1,P_2,\dots,P_{\ell}$, such that every path has a leaf as its source and every vertex (other than the root) belongs to exactly one path as a non-sink vertex. These paths can be created with the following procedure: Initially say that only the root of the tree belongs to a set $X$. Select as $P_1$ any path from a leaf to the root and say that from now on, $X$ also contains the vertices of $P_1$. For $i=2,$ select the path that does not use any vertex of $X$ as a non-sink vertex, starting from an unexplored leaf and ending at a vertex of $X$ and call it $P_i$; add the vertices of $P_i$ to $X$ and repeat for $i=3,\dots,\ell$.

Therefore, $T$ can be represented as a collection of, say $\ell(T)$, 
such paths in a way that each of its vertices (other than the root) appears in that collection as a non-sink vertex exactly once. If we call $w_{P_i}$ the total weight of edges in $P_i$, then the cost of $T$ can be expressed as $\sum_{i\in [\ell(T)]}w_{P_i}$, and thus, the cost of $T'$ can be expressed as
\begin{align*}
    \sum_{i\in [\ell(T)]}&\big(\sum_{u\in N_{P_i}}(max(u)+min(u))-w_{P_i}\big) \leq\\ \sum_{u \in D}&(max(u)+min(u))-k=k'
\end{align*}

% On the other hand, for every pair of successive edges $e_1,e_2$ of $T$ such that the source of $e_1$ is a vertex $u\in D$, there is an edge of weight $max(u)-w(e_1)+min(u)$ in $T'$. Edges of $T'$ that do not correspond to such successive pairs of edges of $T$, are of zero weight. Hence, the total weight of the edges in $T'$ is 
% $\sum_{u\in D}(max(u)-w_u+min(u))\leq k'.$

% $$\sum_{v\in D}((n-1)-d(v)+w_v)=n|D|-|D|-\sum_{v\in D}d(v)+\sum_{v\in D}w_v.$$
% On the other hand, since the first and last edge of every maximal path of $T'$ is zero, $T'$ has a total cost equal to
% $$\sum_{v\in D}(d(v)-w_v+1)=\sum_{v\in D}d(v)-\sum_{v\in D}w_v+|D|.$$
% Finally, observe that the inequality $$ n|D|-|D|-\sum_{v\in D}d(v)+\sum_{v\in D}w_v \geq k$$ implies the inequality $$\sum_{v\in D}d(v)-\sum_{v\in D}w_v+|D| \leq n|D|-k=k'.$$

For the reverse direction, we firstly note that it is without loss of generality to assume that any optimal \textsc{d-mst} of $G'$ contains a path from $\triangledown$ to any special (terminal or non-terminal) vertex, since every special non-terminal vertex $v'$ can be reached from $r'$ at no cost by following the edges $(r',((v',\triangledown),L))$ and $(((v',\triangledown),L),v')$, which definitely exist. The rest of the proof follows from the same arguments presented in the forward direction.

Given an instance $(G'(V',E',w'),r',\hat{V},k')$, the \textsc{d-mst} can be solved in time $O^*(3^{|\hat{V}|})$ by a modification of the classical algorithm for the (undirected) \textsc{minimum steiner tree} problem \cite{dreyfus1971steiner}, where the $O^*$ notation denotes the suppression of factors polynomial in the input size. Further improvements on the running time have also been suggested, as outlined in greater detail in \cite{jones2013parameterized}. By our construction it holds that $|\hat{V}|=|D|$. Hence, \textsc{resolve-delegation} in a time-conscious and confluent manner is solvable in time exponential only in the number of delegating voters of the instance, for t-LD elections of retrospective trust. \hfill$\qed$

\section{Illustrative Examples}
\label{sec:appendix examples}

This section aims to clarify some of the technical terms and concepts used in our work, through simple illustrative examples. The graphs in Figure \ref{fig:ap_graphs} will be used as references throughout this section. For ease of presentation, 
 we have labeled each edge of the graphs in Figure \ref{fig:ap_graphs} with the corresponding time instant in which the edge is available, since in these particular examples, we have assumed that each edge is available for only one time instant. Specifically, we use the label $t$ to denote the time interval $[t,t]$.

\begin{itemize}
    \item {\bf Rooted directed trees.} In $G_1$, we observe a static variant of a temporal directed tree that cannot be said to be rooted at any of its vertices. However, by replacing the edge $(v_1,y_1)$, with its opposite direction, a directed tree rooted at $u_1$ is being produced. 
\item {\bf Time-conscious paths.} Consider $G_2$ to be a temporal graph with a lifespan of $5$, and say that $d_v^t=2$ for every vertex $v$ and for every time-instant $t\in {1,2,\dots,5}$. In this case, the only $\delta$-time-conscious path from $y_2$ to $u_2$ is $(y_2,x_2,w_2,u_2)$. On the other hand, if $d_{y_2}^t=4$ for every time-instant $t\in {1,2,\dots,5}$, then all paths from $y_2$ to $u_2$ except $(y_2,w_2,u_2)$ are $\delta$-time-conscious. 
\item {\bf Confluent vs non-confluent rules.} Suppose $G_3$ represents the unweighted variant of the static graph that arises from a temporal graph that models a t-LD election in which $x_3$ and $y_3$ are casting voters, $w_3$ is an abstaining voter, and $x_3$ is the most preferred representative of $v_3$. Any confluent voting rule should output either $(v_3,x_3)$ or $(v_3,y_3)$, but not both, and should not output $(w_3,y_3)$. However, if the path that is formed by the edges $(u_3,v_3)$ and $(v_3,x_3)$ is an infeasible option (perhaps due to the time-horizon parameters of the voters), a non-confluent voting rule might propose $x_3$ as the representative for $v_3$ and $y_3$ as the representative for $u_3$, via any of the available paths. 
\item {\bf Restless paths.} The notion of $\Delta$-restless paths has been used in the proof of Theorem \ref{2nd-hardness}. In $G_4$, there are two possible paths from $u_4$ to $x_4$, but only the path that goes through $w_4$ is 2-restless. 
\item {\bf Time-respecting paths.} This concept has been crucially utilized in the proof of Theorem \ref{thm:np-hard}, as well as in Procedure \ref{proc2}. In $G_5$, there are three paths of length 2, but only the one that uses the edges $(u_5,v_5)$ and $(v_5,x_5)$ is time-respecting. 
\item {\bf Restless walks vs paths.} The importance of allowing walks instead of paths in the search for restless journeys becomes evident when examining the graph $G_6$. Specifically, the absence of an $1$-restless path from $u_6$ to $w_6$ contrasts with the existence of an $1$-restless walk, by utilizing the cycle between $v_6, x_6, y_6$, highlighting the potential for creating more feasible options with the latter approach.
\end{itemize}

\newpage 
\pagebreak

\begin{figure}[!ht]
% \centering
\fbox{\includegraphics[scale=0.28
]
{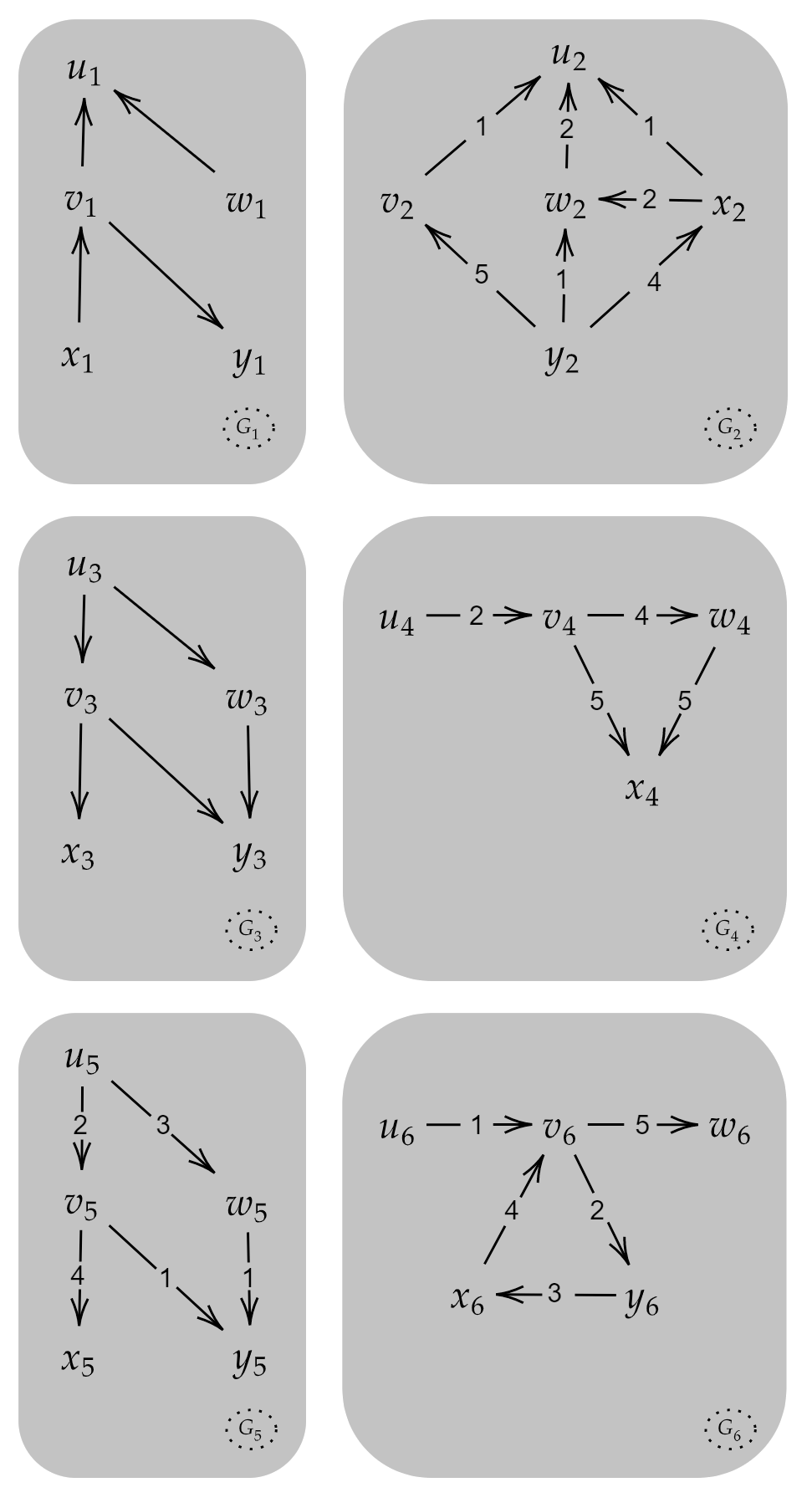}}
% \captionsetup{labelformat=empty,font={small,it}}
\centering 
\caption{Graphs that serve to demonstrate key concepts and terminology used in our work. In the order presented, they correspond to the following concepts: rooted directed trees ($G_1$), time-conscious paths ($G_2$), confluent delegation rules ($G_3$), restless paths ($G_4$), time-respecting paths ($G_5$), and restless walks ($G_6$). An edge labeled $t$ is available only during the time-interval $[t,t]$. }
\label{fig:ap_graphs}
\end{figure}

\end{document}